\documentclass[letterpaper,twocolumn,10pt]{article}
\usepackage{usenix}
\usepackage[available]{usenixbadges}

\usepackage{tikz}
\usepackage{amsmath}

\usepackage{hyperref}
            
\usepackage{graphicx}

\usepackage[utf8]{inputenc} 
\usepackage[T1]{fontenc}    

\usepackage{hyperref}       
\usepackage{xurl}            
\usepackage{booktabs}       
\usepackage{amsfonts}       
\usepackage{nicefrac}       
\usepackage{microtype}      
\usepackage{xcolor}         
\usepackage[many]{tcolorbox}    	
\usepackage{listings}

\usepackage{amsthm}
\usepackage[numbers, sort&compress]{natbib}

\newtheorem{mydef}{Definition}
\newtheorem{theorem}{Theorem}

\newtheorem{prop}{Proposition}
\usepackage{soul}
\usepackage{xcolor}
\usepackage{framed}
\usepackage{amsmath}
\usepackage{amssymb}
\usepackage{mathtools}

\usepackage[inline]{enumitem}
\usepackage{pifont}
\usepackage{wrapfig}
\usepackage{amsmath}
\usepackage{amsfonts}
\usepackage{booktabs} 
\usepackage{xcolor}
\usepackage{comment}
\usepackage{multirow}
\usepackage{amssymb}
\usepackage{mathtools,amsthm}
\usepackage{algorithm}
\usepackage[noend]{algorithmic}
\usepackage{epstopdf}
\usepackage{amsthm}
\usepackage{enumitem}
\usepackage{lipsum}
\usepackage{soul}
\usepackage{subcaption}
\usepackage{framed}
\usepackage{times}

\usepackage[frozencache,cachedir=.]{minted}

\usepackage{tcolorbox}

\usepackage{CJKutf8}
\usepackage{diagbox}
\usepackage{empheq}
\usepackage{cases}
\usepackage{bm}
\usepackage[thinc]{esdiff}

\usepackage{tikz}
\def\checkmark{\tikz\fill[scale=0.4](0,.35) -- (.25,0) -- (1,.7) -- (.25,.15) -- cycle;}

\newcommand\hai[1]{\textcolor{blue}{#1}}

\newcommand{\cmt}[1]{{\color{blue} \footnotesize \# #1}}

\tcbuselibrary{minted}

\begin{document}

\date{}

\title{\Large \bf \textsc{NOIR}: Privacy-Preserving Generation of Code with Open-Source LLMs}

\author{
{\rm Khoa Nguyen$^{1\star}$, Khiem Ton$^{1\star}$, NhatHai Phan$^{1\#}$, Issa Khalil$^{2}$, Khang Tran$^{1\star}$,} \and {\rm Cristian Borcea$^{1}$, Ruoming Jin$^{3}$, Abdallah Khreishah$^{1}$, My T. Thai$^{4}$}\\
$^1$ New Jersey Institute of Technology, $^2$ Hamad Bin Khalifa University, \\ $^{3}$ Kent State University, $^{4}$ University of Florida
\\ Emails: \{nk569, kt477, kt36, borcea, abdallah\}@njit.edu; ikhalil@hbku.edu.qa; \\ rjin1@kent.edu; mythai@cise.ufl.edu \\
$^{\star}$Equal Contributions, $^{\#}$Corresponding Author (Email: phan@njit.edu)
} 

\maketitle

\begin{abstract}
Although boosting software development performance, large language model (LLM)-powered code generation introduces intellectual property and data security risks rooted in the fact that a service provider (cloud) observes a client's prompts and generated code, which can be proprietary in commercial systems. To mitigate this problem, we propose \textsc{NOIR}, the first framework to protect the client's prompts and generated code from the cloud. \textsc{NOIR} uses an encoder and a decoder at the client to encode and send the prompts' embeddings to the cloud to get enriched embeddings from the LLM, which are then decoded to generate the code locally at the client. Since the cloud can use the embeddings to infer the prompt and the generated code, \textsc{NOIR} introduces a new mechanism to achieve indistinguishability, a local differential privacy protection at the token embedding level, in the vocabulary used in the prompts and code, and a data-independent and randomized tokenizer on the client side. These components effectively defend against reconstruction and frequency analysis attacks by an honest-but-curious cloud. Extensive analysis and results using open-source LLMs show that \textsc{NOIR} significantly outperforms existing baselines on benchmarks, including the Evalplus (MBPP and HumanEval, Pass@1 of 76.7 and 77.4), and BigCodeBench (Pass@1 of 38.7, only a 1.77\% drop from the original LLM) under strong privacy against attacks.
\end{abstract}

\section{Introduction}

Commercial LLM-powered code generation tools have greatly boosted developer productivity \cite{copilot}. Yet, over 80\% of companies using cloud-hosted generative AI cite intellectual property (IP) leakage and data security as major concerns, with nearly 45\% reporting unintended data exposure \cite{BiggestRiskGenAI,SecurityRisksCopilot,TooMuchAccess,SEJ}, including real cases of proprietary code leakage \cite{SamsungIncident}. 
These risks are rooted in the prompts a client sends to the cloud-hosted LLMs to get generated code. The prompts and the generated code allow the cloud operators to observe the exact functions and techniques behind commercial systems. Therefore, addressing these concerns and risks is critical given the proliferation of LLM-powered code generation across industry sectors \citep{Market2032}.

\begin{figure}[t]
\centering
\resizebox{1.0\linewidth}{!}{
\includegraphics{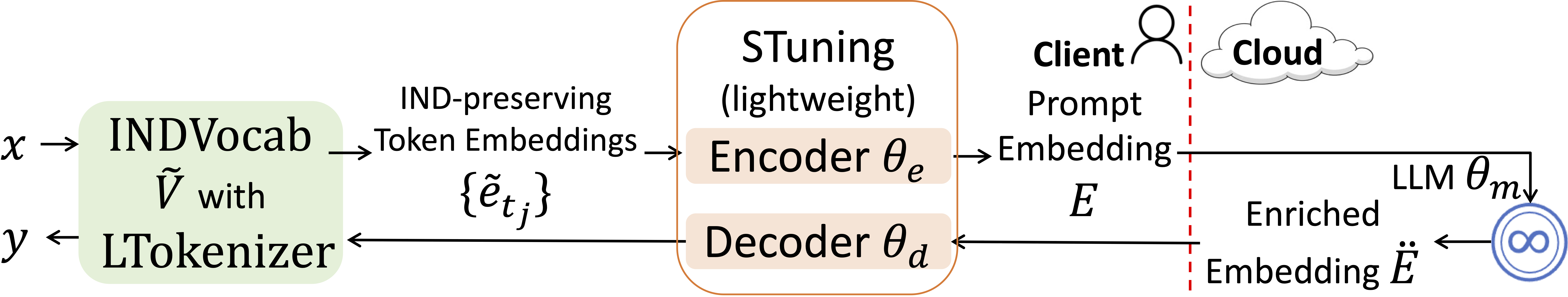}
}
\caption{\textsc{NOIR}: Privacy-preserving Generation of Code.}
\label{XCode}
\end{figure}

\textbf{Private Generation of Code.} A vital solution is protecting both client prompts and generated code from cloud observation. Hosting proprietary LLMs on the client side is impractical, even in compressed form \citep{xiao2023offsitetuning}, as it risks model ownership loss. An alternative is client-side data centers to train, host, and maintain open-source LLMs for code generation, benefiting clients with safety-critical or sensitive data \cite{MediumOpenSourceLLM,OpenSourceLLM,paloaltonet,DigitalTrade}. However, this remains financially unfeasible for most clients and enterprises \citep{GenAIDataCenter,AIDataCenter}. 

\textbf{Prior approaches} include prompt tuning \citep{NEURIPS2023_f26119b4,wu2024privacypreserving,hong2024dpopt} and example ensembling \citep{tang2024privacypreserving} under differential privacy (DP) to protect membership information of data points in the training set; that is, prevent the cloud from inferring whether a data point is used in training prompt-tuning models.
Other methods to protect the client's prompts are to alter tokens \citep{10.1145/3459637.3482281,10.1145/3336191.3371856,li2023privacypreserving} or their embeddings \citep{mai2024splitanddenoise} in every client's prompt to achieve \textit{metric} DP (i.e., $d_x$-privacy \citep{10.1007/978-3-642-39077-7_5}).
These state-of-the-art (SoTA) approaches focus on either classification tasks \citep{10.1145/3459637.3482281,10.1145/3336191.3371856,li2023privacypreserving,mai2024splitanddenoise} instead of code generation tasks as pointed out in \citep{mai2024splitanddenoise} or protecting the training data for prompt tuning \citep{NEURIPS2023_f26119b4,wu2024privacypreserving,hong2024dpopt,tang2024privacypreserving} without protecting the content of client's prompts and generated outcomes. 

Therefore, a novel approach is desired for the private generation of code, where the sweet-spot is balanced between rigorous privacy-preserving guarantees for prompts and generated code, cost effectiveness, and high model performance.

\textbf{Challenge.} Unlike generating general text, the main challenge in private code generation is that small changes in clients' prompts for privacy protection—at the token or its embedding level—can severely degrade generated code functionality, as LLMs are highly sensitive to such changes. Ensuring good generated code functionality often weakens $d_x$-privacy protection in SoTA methods \citep{10.1145/3459637.3482281,10.1145/3336191.3371856,li2023privacypreserving,mai2024splitanddenoise}, leaving prompts and generated code vulnerable to reconstruction attacks \citep{morris-etal-2023-text,chen2024unveilingvulnerabilityprivatefinetuning} by privacy-untrusted clouds. In addition, we prioritize code generation over general text due to its severe risks to IP and security. Leaks of proprietary algorithms, sensitive system code implementation, and trade secrets create a massive blind spot for zero-day vulnerabilities, IP theft, and supply chain compromises causing security, legal, and economic damages far greater than in general text generation \cite{sourcecodeleaks}.

\textbf{Contributions.} This paper proposes \textsc{NOIR} (Figure \ref{XCode}), the first framework to protect client prompts and generated code from cloud observation using an open-source LLM. Leveraging the competitive performance of open-source LLMs \cite{chen2021evaluating,zhuo2024bigcodebench}, service providers can offer \textsc{NOIR} to clients and enterprises demanding strong privacy guarantees \cite{MediumOpenSourceLLM,OpenSourceLLM,paloaltonet,DigitalTrade}. \textsc{NOIR} uses open-source LLMs to provide privacy guarantees to clients, and it complements closed-source LLM-based code generation tools, where it is difficult to provide privacy protection without exposing closed-source models to clients.

Given an open-source LLM, \textsc{NOIR} splits it into three parts: 1) \textbf{encoder} (first few attention blocks), 2) \textbf{middle part} (most attention blocks), and 3) \textbf{decoder} (last few attention blocks).
On the client side, \textsc{NOIR} includes the encoder, decoder, a fine-tuning method (\textsc{STuning}), and an indistinguishability (IND)-preserving vocabulary (\textsc{INDVocab}), which provides a local differential privacy (LDP) protection at the token embedding level, associated with a local randomized tokenizer (\textsc{LTokenizer}). Privacy protection arises from running the encoder and decoder locally: prompts are encoded before being sent to the cloud-hosted (middle part of the) LLM, which enriches prompt embeddings using its latent knowledge. The enriched embeddings enhance model performance on a range of downstream tasks because the middle part of LLMs can capture even richer representations from large datasets compared with the final layers of LLMs \cite{skean2025layer}. These embeddings are returned to the client’s decoder to generate code. This design avoids sending raw prompts or code to the cloud.


To optimize the encoder and decoder to client tasks, we develop \textsc{STuning}, a cost and performance-effective split learning approach \citep{DBLP:journals/corr/abs-1812-00564,10529950,10.1145/3460120.3485259}, to fine-tune them locally with client datasets, mitigating utility loss from privacy protection. Since the encoder and decoder are lightweight compared with the cloud-hosted model, \textsc{STuning} is cost-effective. Optionally, the cloud may fine-tune a low-rank adaptation (LoRA \citep{hu2022lora}) of the hosted model for improved performance.

\begin{table*}[t]
\centering
\caption{A Summary of Differences between \textsc{NOIR} and Related Works.}
\resizebox{0.885\textwidth}{!}{
\begin{tabular}{lcccclr}
\toprule
 & Generation & Denoise & Fine-tuning & A Token Embedding in Multiple Prompts & Tokens in a Prompt & Privacy Guarantee \\ 
\midrule
TokEmbPriv \citep{10.1145/3459637.3482281} & $\mathsf{x}$ & $\mathsf{x}$ & $\mathsf{x}$ & changed, inconsistent & changed, inconsistent & $d_x$-privacy \\ 
Text2Text \citep{10.1145/3336191.3371856} & $\mathsf{x}$ & $\mathsf{x}$ & $\mathsf{x}$ & changed, inconsistent & changed, inconsistent & $d_x$-privacy \\ 
RAPT \citep{li2023privacypreserving} & $\mathsf{x}$ & \checkmark & \checkmark & changed, inconsistent & changed, inconsistent & $d_x$-privacy \\ 
Split-and-Denoise \citep{mai2024splitanddenoise} & $\mathsf{x}$ & \checkmark & $\mathsf{x}$ & changed, inconsistent & changed, inconsistent & $d_x$-privacy \\ 
\textbf{\textsc{NOIR} (ours)} & {\checkmark} & ${\mathsf{x}}$ & {\checkmark} & {unchanged, consistent} & {unchanged, consistent} & {$\epsilon$-IND} \\
\bottomrule
\end{tabular}
}
\label{RelatedWorkTable}
\end{table*}

To prevent the cloud from inferring sensitive and proprietary content from the prompts and the generated code via prompt embeddings and back-propagated gradients under SoTA reconstructing attacks \citep{morris-etal-2023-text,chen2024unveilingvulnerabilityprivatefinetuning} in both the \textsc{STuning} and inference phases, we propose a client-side privacy mechanism with two components: \textsc{INDVocab} and \textsc{LTokenizer}. \textsc{INDVocab} adaptively randomizes the token embeddings in the vocabulary making them indistinguishable to the cloud so the probability of the cloud inferring ground-truth tokens, prompts, and code given the prompt embeddings is upper-bounded, while minimizing injected randomness for better utility.
To protect the one-hot vectors of tokens used in the client tokenizer, which the cloud can exploit during \textsc{STuning} to reconstruct the client's data from the back-propagated gradients \citep{chen2024unveilingvulnerabilityprivatefinetuning}, \textsc{NOIR} develops \textsc{LTokenizer}, which uniformly assigns every token and its IND-preserving embedding to a random index in the \textsc{INDVocab}. This data-independent tokenizer remains secret (i.e., no extra privacy cost), misleading the cloud to observe meaningless tokens from the back-propagated gradients.

\textsc{LTokenizer} and \textsc{INDVocab} fully protect client prompts and output code during fine-tuning and inference. \textsc{INDVocab} keeps tokens in prompts, (prompt) instructions, (prompt) templates, and code intact while randomizing token embeddings only once with negligible noise. Thus, it maintains the correlation among prompts, instructions, templates, and code while providing indistinguishability protection to every token, given its token embedding. This correlation enhances \textsc{STuning} by mitigating utility drops from encoder-(cloud-hosted) LLM-decoder misalignment when using the \textsc{INDVocab} and \textsc{LTokenizer}. Consequently, \textsc{NOIR} generates code with high functionality under strong client-side IND protection.

Extensive experiments with LLMs (CodeLlama-7B, CodeQwen1.5-7B-Chat, Llama3-8B-instruct) on benchmarks (Evalplus: MBPP \citep{austin2021program}, HumanEval \citep{chen2021evaluating}; BigCodeBench \citep{zhuo2024bigcodebench}) show that \textsc{NOIR} achieves a Pass@1 scores of 76.7 and 77.4 on MBPP \citep{austin2021program} and HumanEval \citep{chen2021evaluating} and 38.7 on BigCodeBench (a marginal drop of 1.77\% from the original LLM) while significantly outperforming SoTA baselines (T2T \citep{10.1145/3459637.3482281,10.1145/3336191.3371856,li2023privacypreserving}, SnD \citep{mai2024splitanddenoise}) against reconstruction and frequency analysis attacks \citep{morris-etal-2023-text,chen2024unveilingvulnerabilityprivatefinetuning}. Regarding cost effectiveness, \textsc{NOIR} reduces client inference/fine-tuning costs by $\backsim10$x compared with local LLM hosting, thanks to its lightweight encoder–decoder (1–4 attention blocks). GPU memory use, fine-tuning time, and equivalent AWS hosting costs grow sub-linearly with dataset size, making fine-tuning scalable on larger datasets without sharply increasing client communication costs.

To show practicality, we open-source \textsc{NOIR} ({\color{blue}\url{https://tinyurl.com/NOIR-Artifact}}) based on Qwen2.5-Coder-32B-Instruct and provide an application programming interface (API) via a privacy-preserving coding agent, accessible through a web service ({\color{blue}\url{https://noir.oppyai.com}}) and a Visual Studio (VS) extension ({\color{blue}\url{https://tinyurl.com/NOIR-Artifact}}), integrated into the development pipeline.

\section{Related Work}
\label{Related Work}







\textbf{Privacy-preserving Prompts.} SnD \citep{mai2024splitanddenoise} is the most recent related work. The client injects independent draws of Laplace noise into the token embeddings of every prompt to achieve $d_x$-privacy \citep{10.1007/978-3-642-39077-7_5} (relaxed LDP) before sending the embeddings to the cloud. Then, it receives output embeddings in return and denoises them for downstream tasks. 
The common point among other studies \citep{10.1145/3459637.3482281,10.1145/3336191.3371856,li2023privacypreserving} is achieving $d_x$-privacy by replacing tokens with randomized
tokens, after injecting noise into token embeddings of a prompt.

Table \ref{RelatedWorkTable} summarizes key differences between \textsc{NOIR} and prior methods.
These methods were designed for classification, and extending them to sequence-to-sequence generation is difficult, as errors in previous token prediction will compound the deviation of the following tokens and substantially degrade performance \citep{mai2024splitanddenoise}. We observe similar issues in our study. Therefore, instead of focusing on achieving $d_x$-privacy at the prompt level as in SoTA, \textsc{NOIR} enables private code generation—a harder task requiring protection of both prompts and generated code while preserving functionality by achieving indistinguishability of the vocabulary (\textsc{INDVocab}) and the \textsc{LTokenizer}. NOIR adaptively randomizes token embeddings only once, making them indistinguishable to the cloud, and keeping them consistent in \textsc{INDVocab} used across prompts and code. 
This reduces the randomness injected while maintaining tokens' essential correlation in the client's data; hence, \textsc{NOIR} allows the client to fine-tune the encoder and decoder locally for code generation.

In the scenario where users/organizations need to provide data to an enterprise (client) in \textsc{NOIR}’s setting, they must trust the enterprise as a data curator in real-world deployments. This setting is different from classical LDP protection, in which the user/organization's data, i.e., prompts and generated code, is protected from the data curator \cite{10.1145/2660267.2660348}, i.e., the client in \textsc{NOIR}. If the users/organizations have sufficient resources and data to host the encoder/decoder in practice, they can skip the enterprise (client) and directly use \textsc{NOIR}.

\textbf{Privacy-preserving Prompt Tuning (P3T).} The most recent P3T \citep{NEURIPS2023_f26119b4,wu2024privacypreserving,hong2024dpopt} is DP-OPT \citep{hong2024dpopt}, which fine-tunes a local LLM for ensemble prompt engineering, with DP to protect the membership information of prompts in the local LLM's training set. 
\textbf{This objective is different from protecting the content of the client's prompts and generated code} in both fine-tuning and inference phases of \textsc{NOIR}.

\section{Preliminaries}
\label{Preliminaries}





\textbf{LLM-based Code Generation.}
An LLM uses a vocabulary of tokens (human-readable words) and their token embeddings: $V=\{t, e_t\}$, where $|V|$ is the vocabulary size and $e_t \in \mathbb{R}^m$ with $m$ features. A prompt $x$ is a sequence of tokens from $V$: $x = \{t_j\}_{j =1}^{|x|}$ where $t_j \in V$ and $|x|$ is the number of tokens in $x$. Combined with a task instruction $\pi$ in a template $\mathcal{T}$\footnote{For instance, $\pi$: ``You are an expert in Python:'' and $\mathcal{T}$: ``Complete the request. \#\#\# Instruction: \{instruction $\pi$\} \{prompt $x$\} \#\#\# Response:''}, this guides the LLM to generate output code $y$, modeled as $P^h[y|\mathcal{T}(x, \pi)]$, where the higher the temperature $h$ ($\geq 0$), the more diverse the generated code. For simplicity, we fix $\pi$ and $\mathcal{T}$, treating $x$ as the input. The prompt $x$ is then represented as a sequence of token embeddings $\{e_{t_j}\}_{j = 1}^{|x|}$ fed to the LLM: $P^h[y|\{e_{t_j}\}_{j = 1}^{|x|}]$.

\textbf{Tokenization} consists of the first and the last steps of text processing and modeling in LLMs \citep{MistralAITokenizer}. A tokenizer breaks down text into tokens and assigns each token $t \in V$ a unique numerical index represented by a one-hot vector $v_t \in \mathbb{I}^{|V|}$ without affecting the generality and correctness.
Given the input $\{e_{t_j}\}_{j = 1}^{|x|}$, the LLM iteratively generates the one-hot vector $v_t$ of the next token $t$ in the output code $y$. Finally, the tokenizer detokenizes the one-hot vectors $\{v_t\}_{t \in y}$ back to human-readable text by mapping those vectors to their corresponding tokens using the vocabulary $V$.

\textbf{Differential Privacy.} DP \citep{dwork2014} is widely-used for data privacy, and LDP \citep{4690986,6686179} particularly protects the values of data inputs against an untrusted data curator. In the classical LDP definition, an LDP-preserving mechanism produces similar output distributions, preventing the curator from distinguishing the outcomes of the data inputs.
\begin{mydef}{$\epsilon$-LDP \cite{6686179}.} A randomized algorithm $\mathcal{M}$ fulfills $\epsilon$-LDP, if for any two inputs $x$ and $x'$, and for all possible outputs $\mathcal{O}$ of $\mathcal{M}$ $\big(\mathcal{O} \in Range(\mathcal{M})\big)$, we have:
$Pr[\mathcal{M}(x) = \mathcal{O}] \leq e^{\epsilon} Pr[\mathcal{M}(x') = \mathcal{O}]$, where $\epsilon$ is a privacy budget.
\label{LDP}
\end{mydef}
\noindent A smaller $\epsilon$ enforces a stronger privacy guarantee controlling the difference of the distributions induced by $x$ and $x'$. Our setting is different from the classical LDP, where users/organizations with limited resources must trust the client to directly use \textsc{NOIR} (Section \ref{Related Work}).

\textbf{Split Learning} \citep{DBLP:journals/corr/abs-1812-00564} enables distributed learning by decomposing a neural network, e.g., an LLM in \textsc{NOIR}, into non-overlapping client and cloud segments. In \textsc{NOIR}, we consider a client (e.g., a small enterprise with limited resources) and a cloud, without loss of generality. Typical split learning frameworks use two architectures \citep{10529950,10.1145/3460120.3485259}; \textsc{NOIR} adopts the setting that protects both the input prompt $x$ and output code $y$: the client holds the first and last several attention blocks, while the cloud hosts the middle blocks of the LLM. 


\textbf{Reconstruction Attacks (RAs).} RAs aim to recover input text from its embedding \citep{ni2021largedualencodersgeneralizable,li2023towards,10.1145/3372297.3417270,li-etal-2023-sentence} or both input and output text from their associated back-propagated gradient \citep{chen2024unveilingvulnerabilityprivatefinetuning}.
Vec2Text \citep{morris-etal-2023-text}, the most advanced embedding-based RA, has base and refining steps: a trained conditional language model converts the embedding to a text corpus (the ``base hypothesis''), then recursively re-embeds and corrects the base hypothesis to increase cosine similarity with the original text embedding. 
The Vec2Text model is trained on this generated data. BiSR \citep{chen2024unveilingvulnerabilityprivatefinetuning}, the latest gradient-based RA, initializes a dummy label and iteratively improves it through backward gradient matching and forward embedding matching.

\begin{algorithm}[t]
\footnotesize
\caption{\textsc{NOIR}: Private Generation of Code}
\label{CodeX - Psuedo Code}
\begin{algorithmic}[1]
\STATE \textbf{Input}: Vocabulary $V$, Encoder $\theta_e$, Decoder $\theta_d$, Cloud-hosted Model $\theta_m$, IND budget $\epsilon$, Training Data $D = \{x, y\}$, learning hyper-parameter $\gamma$
\STATE \textbf{Output}: $\epsilon$-IND-preserving $\tilde{V}$, \textsc{LTokenizer}, $\theta_e, \theta_d$
\STATE \textbf{Def Client($V, \theta_e, \theta_d, D$)} \\
\STATE \textbf{Initialize} $\tilde{V}$ as a copy of the vocabulary $V$ \cmt{Creating \textsc{INDVocab}} \\
\FOR{token $t \in V$}
    \FOR{$i^\text{th}$-feature of token embedding $e_t$ of $t$}
        \STATE \textbf{Compute} $\beta_i$ using Theorem \ref{theorem-beta-bound} \\
        \STATE \textbf{Assign} a value to $\tilde{e}^i_t$ with ARR using Eq. \ref{ARR} with $\beta_i$
    \ENDFOR
    \STATE \textbf{Assign} $\{t, \tilde{e}_t\}$ a random, unique index in the client's tokenizer \cmt{\textsc{LTokenizer}}
\ENDFOR
\FOR{$\text{round} \in [1,T]$ \cmt{\textsc{STuning} with \textsc{LTokenizer}, \textsc{INDVocab}}}
    \STATE \textbf{Sampling} a batch $B$ of data points $\{x, y\} \in D$ \\
    \STATE $\{\mathcal{E} \leftarrow Enc\big(\mathcal{T}(x, \pi), \theta_e\big)\}_{x \in B}$ \cmt{Get prompt embeddings using $\tilde{V}$} \\
    \STATE $\ddot{\mathcal{E}} \leftarrow \textbf{Cloud}(\{\mathcal{E}\})$ \cmt{Send $\{\mathcal{E}\}$ to the cloud and get $\ddot{\mathcal{E}}$}
    \STATE $\theta_d \leftarrow \theta_d - \lambda \nabla \mathcal{L}(\theta_d)$ \cmt{The client fine-tunes $\theta_d$ using $\ddot{\mathcal{E}}$ and $y$} \\
    \STATE $\theta^{LoRA}_m \leftarrow \theta^{LoRA}_m - \lambda \nabla \mathcal{L}(\theta^{LoRA}_m)$ \cmt{The cloud fine-tunes $\theta_m$'s LoRA} \\
    \STATE $\theta_e \leftarrow \theta_e - \lambda \nabla \mathcal{L}(\theta_e)$ \cmt{The client fine-tunes $\theta_e$ using $\mathcal{L}$}
\ENDFOR
\STATE \textbf{Def Cloud($\{\mathcal{E}\}$)} \\
\STATE \textbf{\ \ \ Return} $\{f_{emb}(\mathcal{E}, \theta_m)\}$ \cmt{Extract and return enriched embeddings}
\end{algorithmic} 
\end{algorithm}

\section{\textsc{NOIR}: Overview, Feasibility, and Threats}
\label{X-Code}

\textbf{Overview.} IP leakage and data security motivate \textsc{NOIR}, designed to enforce \textbf{CPC} (data \underline{C}onfidentiality, information \underline{P}rivacy, and code \underline{C}onfidentiality) constraints.
\textbf{Raw Data Confidentiality:} Client data—built from open/private sources, prompts, or generated code during inference—must not be shared with the (privacy-untrusted) cloud.
\textbf{Encoded-Information Privacy:} Sensitive content in client prompts and code, across training and inference, must remain uninterpretable to the cloud.
\textbf{Code Confidentiality:} Code syntax, semantics, and functionality in client data and generated outputs must not be reconstructable by the cloud.

\textbf{General Design.} To accommodate the CPC constraints in \textsc{NOIR} (Figure \ref{XCode}, Alg. \ref{CodeX - Psuedo Code}),
the cloud decomposes an open-source LLM $\theta$ into three parts: an encoder with the first few attention blocks $\theta_e$, a decoder with the last few attention blocks $\theta_d$, and the remaining middle blocks $\theta_m$. The lightweight encoder and the decoder associated with the vocabulary $V$ and its tokenizer are known to the client so that the client can use them at a negligible cost. Meanwhile, the cloud hosts the remaining blocks $\theta_m$. The pair of an encoder and a decoder satisfies the data confidentiality constraint. By integrating \textsc{INDVocab} and \textsc{LTokenizer}, the design satisfies information privacy and code confidentiality constraints. 

\textbf{Feasibility.} Our setting is feasible because: \textbf{(1)} Most LLM-providers (e.g., OpenAI, Meta, Google, DeepSeek, Alibaba, etc.) release open-source versions of their closed-source models. Service-providers thus have an incentive to offer open-source encoders and decoders to clients concerned with IP leakage and data security \cite{BiggestRiskGenAI,SecurityRisksCopilot,TooMuchAccess,SEJ,MediumOpenSourceLLM,OpenSourceLLM,paloaltonet,DigitalTrade}, attracting new customers without affecting existing closed-source offerings; \textbf{(2)} Operation costs—hosting, fine-tuning, and inference—are notably reduced with lightweight encoders and decoders, making deployment practical for most clients and enterprises; and \textbf{(3)} Open-source LLMs now achieve highly competitive performance \cite{chen2021evaluating,zhuo2024bigcodebench}. Hence, open-source LLMs align well with the design of \textsc{NOIR}.

The client uses the open-source vocabulary $V$ and the tokenizer to represent every prompt $x$ as a sequence of token embeddings $\{e_{t_j}\}_{j = 1}^{|x|}$. The encoder $\theta_e$ extracts a prompt embedding $\mathcal{E} = Enc\big(\{e_{t_j}\}_{j = 1}^{|x|}, \theta_e\big)$, which is sent to the cloud-hosted model $\theta_m$ to obtain an enriched embedding $\ddot{\mathcal{E}} = f_{emb}(\mathcal{E}, \theta_m)$. The client feeds $\ddot{\mathcal{E}}$ to the decoder $\theta_d$ to generate the output code $y = Dec(\ddot{\mathcal{E}}, \theta_d)$. The encoder and decoder are fine-tuned (\textsc{STuning}, Section \ref{STuning}) for local tasks using curated data from open and/or the client's private sources. The cloud can either \textbf{1) keep $\theta_m$ fixed} to reduce computation complexity or \textbf{2) fine-tune a lightweight LoRA} \citep{hu2022lora} of the middle block $\theta_m$ with the client for improved utility at marginal cost.

\subsection{Threat Model}
\label{Threat Model}

\begin{figure}[h]
\centering
\resizebox{0.99\columnwidth}{!}{
\includegraphics{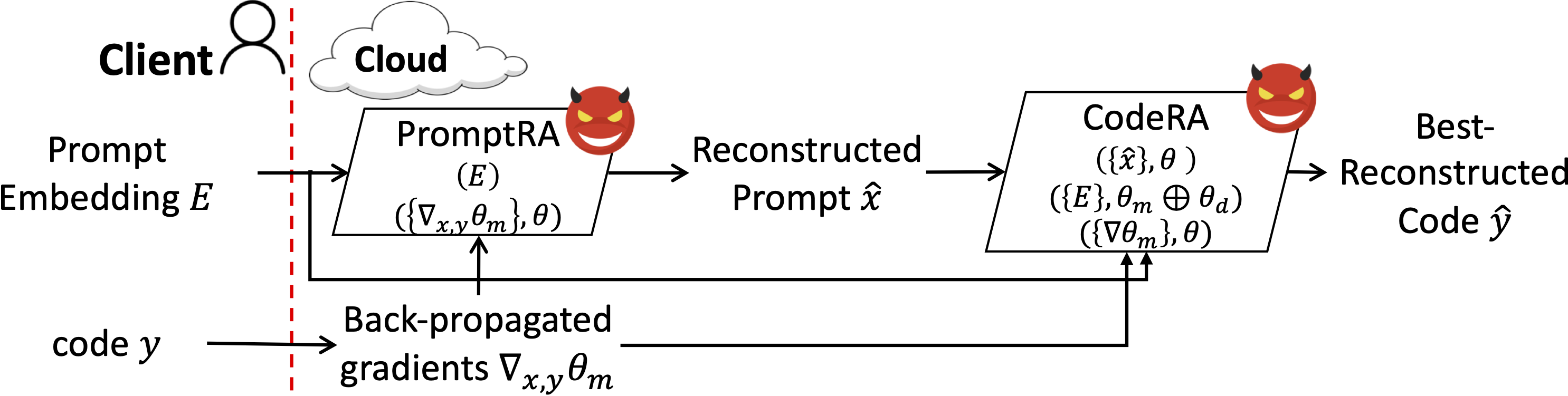}
}
\caption{Threat Model of an Honest-but-Curious Cloud.}
\label{ThreatModel}
\end{figure}

In a \textbf{defense-free environment}, the client uses the open-source vocabulary and tokenizer during fine-tuning and inference. Figure \ref{ThreatModel} shows the attack surface in \textsc{NOIR}. 
The honest-but-curious cloud seeks to reconstruct prompts, code used in the client's fine-tuning, and code generated in the inference phase.
To reconstruct the prompts (PromptRA), the cloud applies SoTA RAs \citep{morris-etal-2023-text} on the client's prompt embeddings, then feeds them to the LLM model $\theta$ to generate reconstructed code $\hat{y}$ (CodeRA). Due to noisy prompt embeddings, the reconstructed prompts are typically imperfect, so the cloud alternatively feeds prompt embeddings directly into $\theta_m$ concatenated with the original decoder $\theta_d$.  
During fine-tuning, the cloud can also feed the back-propagated gradient of each training sample $\{x, y\}$ through the middle block $\theta_m$, denoted as $\nabla_{x,y} \theta_m$, into BiSR \citep{chen2024unveilingvulnerabilityprivatefinetuning} to reconstruct the prompts $x$ and their output code $y$ in the client's training set $D$. We evaluate attack performance using the best-reconstructed code in both methods, considering maximal information leak in each metric, including information privacy, code confidentiality, and code functionality. 

We exclude attacks involving compromised employees in the client's organization. These employees can collude with the cloud for the cloud to query the client's encoder or disclose the input prompts $x$, the output code $y$, the \textsc{INDVocab}, the \textsc{LTokenizer}, the fine-tuned encoder and decoder, and their associated gradients to the cloud. Such insiders can gather all raw data during training and inference (out of \textsc{NOIR}'s scope). \textbf{Token sequence-based frequency analysis attacks}, e.g., the codebook attack \cite{Biryukov2011}, \textbf{do not apply to \textsc{NOIR}} since: \textbf{(1)} the encoder's attention mechanism yields varied token embeddings for the same input token given different token positions; and \textbf{(2)} The cloud cannot query the client's encoder or access raw inputs. Therefore, \textbf{adaptive attacks} (manipulating prompts to observe output changes) \textbf{are not applicable}.

\begin{figure*}[t] 
\captionsetup[subfigure]{justification=centering}
    \centering
      \begin{subfigure}[t]{0.495\textwidth}
        \centering
\begin{tcblisting}{colback=white,colframe=black,listing only,listing engine=minted,minted language=text,minted options={breaklines,breaksymbolleft=,fontsize=\tiny},  left=1pt,
  right=1pt,
  top=1pt,
  bottom=6pt,
  boxsep=1pt,
}
Below is an instruction that describes a task. Complete the request.
### Instruction:
You are an expert Python programmer. Write a function to find the first duplicate element in a given array of integers. Your code should pass these tests:
['assert find_first_duplicate(([1, 2, 3, 4, 4, 5]))==4']
### Response: 
\end{tcblisting}
       \caption{A Client's Prompt}
       \vspace{0.3cm}
       \label{REPrompt}
       \end{subfigure}
       \hfill
\begin{subfigure}[t]{0.495\textwidth}
      \centering
\begin{tcblisting}{colback=white,colframe=black,listing only,listing engine=minted,minted language=python,minted options={breaklines,breaksymbolleft=,fontsize=\tiny},  left=1pt,
  right=1pt,
  top=1pt,
  bottom=1pt,
  boxsep=1pt,
}
def find_first_duplicate(arr):
    seen = set()
    for num in arr:
        if num in seen:
            return num
        seen.add(num)
    return -1 
\end{tcblisting} 
    \caption{\textbf{The Client's} Generated Code \textbf{with \textsc{NOIR}}}
    \vspace{0.3cm}
        \label{RECodeX}
    \end{subfigure}  \\
      \begin{subfigure}[t]{0.55\textwidth}
      \centering
\begin{tcblisting}{colback=white,colframe=black,listing only,listing engine=minted,minted language=python,minted options={breaklines,breaksymbolleft=,fontsize=\tiny},  left=1pt,
  right=1pt,
  top=1pt,
  bottom=1pt,
  boxsep=1pt,
}
def find_first_duplicate(nums):
  for i in range(len(nums)):
    if nums[abs(nums[i])] > 0:
      return abs(nums[i])
    else:
      nums[abs(nums[i])] = -nums[abs(nums[i])]
  return -1
\end{tcblisting} 
        \caption{\textbf{The Cloud's} Reconstructed Code \textbf{in a Defense-free Environment}}
        \label{RECloudNoDefense}
    \end{subfigure} \hfill
      \begin{subfigure}[t]{0.44\textwidth}
      \centering
\begin{CJK*}{UTF8}{gbsn}
\begin{tcblisting}{colback=white,colframe=black,listing only,listing engine=minted,minted language=text, minted options={breaklines,breakanywhere=true,breaksymbolleft=, fontsize=\tiny},  left=1pt,
  right=1pt,
  top=1pt,
  bottom=1pt,
  boxsep=1pt,}
def find\n\n\n\n\n\n\n\n\n\n\n\n\n\n\n年年\n\nsum\n\n CONSTRAINT\n\nconditionconditionconditioncondition Mall\ncondition Mallspring Mallspringliedconditionliedliedcondition Mall Mallconditionlied lied Mall Mall Mall Malllied',\\lied',\\ Mallliedslantangledspring春天的春天的EEE',\\$'$'hall gammagammagammagammagammagamma Aroundgamma斥t']routejanroute SpringsJan net']joint\nboot用用 Hung用',\\ Hung Hung Hung Hung Hungfluid',\\979nab年后gamma',\\',\\相识0']',\\220220因素191春天的
\end{tcblisting} 
\end{CJK*}

    \caption{\textbf{The Cloud's} Reconstructed Code \textbf{under \textsc{NOIR}}}
        \label{RECloudCodeX}
    \end{subfigure}
    \caption{An Example of the Threat Model.}
    \label{fig:running example}  
\end{figure*}

\subsection{\textsc{NOIR}'s Defense}

To protect the client's prompts and code against RAs, \textsc{NOIR} has two key components. First, \textsc{INDVocab} replaces the vocabulary $V$ with an IND-preserving $\tilde{V} = \{t, \tilde{e}_t\}$, where every $\tilde{e}_t$ is an IND-preserving token embedding, derived from the original token embedding $e_t$ (Section \ref{LDP-preserving Vocabulary}). Hence, a prompt $x$ is represented by a sequence of IND-preserving token embeddings $\{\tilde{e}_{t_j}\}_{j = 1}^{|x|}$ defending against PromptRA. Second, \textsc{NOIR} develops a local tokenizer (\textsc{LTokenizer}), which uniformly assigns every token and its embedding to a random index in the tokenizer, on the client side to defend against CodeRA. The client keeps this data-independently and randomized tokenizer (no extra privacy cost) secret from the cloud, misleading the cloud's attacks to reconstruct meaningless tokens in the output code $y$ of the client's training data $D$.

At the inference phase, the client uses the \textsc{INDVocab} $\tilde{V}$ associated with the \textsc{LTokenizer}, the fine-tuned encoder $\theta_e^*$, and the decoder $\theta_d^*$ to generate code for its prompts. \textbf{The client does not share the fine-tuned $\theta_e^*$ and $\theta_d^*$ with the cloud, and the cloud does not share the LoRA of the middle block $\theta_m$ during the fine-tuning and inference phases} to maintain their model ownership. \textsc{NOIR} maintains the CPC constraints, incentivizing the client to use LLM-based code generation under IP and data security protection and broadening the adoption of the cloud's service.

\textbf{Example.} 
Given the prompt in Figure \ref{REPrompt}, the cloud reconstructs the client's generated code with a clear gist in the defense-free environment (Figure \ref{RECloudNoDefense}). 
On the contrary, the cloud reconstructs meaningless code under \textsc{INDVocab} and \textsc{LTokenizer} (Figure \ref{RECloudCodeX}); meanwhile, the client enjoys its desired code with \textsc{NOIR} on their side privately (Figure \ref{RECodeX}).
Next, we describe reconstruction attacks as security games between the client and the cloud.

\section{Reconstruction Attacks}
\label{Reconstruction Attacks}

\subsection{Prompt Reconstruction Attack}
\label{Prompt Reconstruction Attack}

In PromptRA, the cloud feeds each prompt embedding $\mathcal{E}$ into Vec2Text \citep{morris-etal-2023-text} or feeds the back-propagated gradients $\nabla_{x, y} \theta_m$ to BiSR \citep{chen2024unveilingvulnerabilityprivatefinetuning}, the SoTA RAs, to generate $\hat{x}$ approximating the ground-truth prompt $x$, formulated as $PromptRA: (\{\mathcal{E}\}, \theta_a) \cup (\{\nabla_{x, y} \theta_m\}, \theta) \rightarrow \hat{x}$, where $\theta_a$ are Vec2Text's pre-trained parameters. The closer $\hat{x}$ is to $x$, the stronger the attack and the greater the information privacy leakage.
In the security game, the cloud may request the client's prompt embeddings $\mathcal{E}$ and their associated back-propagated gradients $\nabla_{x,y} \theta_m$ from the training set $D$ at any time during local fine-tuning.
We denote the sets of all prompt embeddings as $\mathbb{E}$ and their reconstructed prompts as $\hat{X}$.

We evaluate attack success using the well-known token-level metrics Bleu \citep{papineni-etal-2002-bleu} and Rouge \citep{lin-2004-rouge} scores, which measure similarity between the reconstructed prompt $\hat{x}$ and the original $x$. 
The cloud wins the game for a prompt $x \in D$ if it returns $\hat{x} \in \hat{X}$ with a clear gist: either $\max\{Bleu(\hat{x}, x)\}_{\hat{x} \in \hat{X}} \geq \rho_{b}$, where $\rho_b = 20$ \citep{BleuScoreRange}, or $\max\{Rouge(\hat{x}, x)\}_{\hat{x} \in \hat{X}} \geq \rho_r$, where $\rho_r = 0.4$ \citep{ROUGEScoreRange}.
The overall attack success rate (ASR) is the cloud's average winning rate over all prompts in $D$ during the client's fine-tuning. During inference, the cloud requests all prompt embeddings in the test set $D_{test}$ once, since multiple requests make no difference.
By default, we use uni-gram Bleu and Rouge, which yield the best results compared with bi- and longer-grams; otherwise, the specific $n$-grams is noted.

This study uses the ASR thresholds ($\rho_b = 20$, $\rho_r = 0.4$) to evaluate the effectiveness of our model against RAs and baselines. In real-world deployments, the client can ignore these thresholds, since the client will receive a security report showing the cloud’s reconstructed prompts and outcomes for each prompt. Hence, it can determine whether sensitive information has been leaked. The thresholds balance practical privacy protection with realistic evaluation of reconstruction risks. The cloud's ability to infer exploitable information from encoded embeddings or gradients is quantitatively limited, whereas distinguishing between useful reconstructions and noise or meaningless outputs (lower threshold values) poses little risk. Overly strict thresholds can lead to a high false-positive rate. \textsc{NOIR}'s thresholds avoid false positives by considering exploitable reconstructions as the only alarming threat, which is more realistic and actionable for clients.

\begin{table*}[t]
\centering
\caption{PromptRA and CodeRA in a Defense-free Environment.}
\resizebox{0.58\textwidth}{!}{
\begin{tabular}{lcccccccc}
\toprule
\multirow{3}{*}{\textbf{PromptRA}} & \multicolumn{4}{c}{Training Data $D$} & \multicolumn{4}{c}{Testing Data $D_{test}$} \\ \cmidrule(rl){2-5} \cmidrule(rl){6-9}
 & Bleu & Rouge & $ASR^{priv}_{x, D}$ &$CRT_x$ & Bleu & Rouge & $ASR^{priv}_{x, D_{test}}$ & $CRT_x$ \\
\midrule
MBPP (no test cases in prompts, no fine-tuning) & 34.17 & 0.66 & \textbf{0.96} &\textbf{0.32} & 35.53 & 0.67 & \textbf{0.96} &  \textbf{0.32}\\ 
MBPP (no test cases in prompts, fine-tuning) & 33.92 & 0.66 & \textbf{0.96}& \textbf{0.32} &34.08 & 0.66 & \textbf{0.95}& \textbf{0.32}\\ 
MBPP (no fine-tuning) & 11.25 & 0.40 & \textbf{0.50} & \textbf{0.1} &  11.49 & 0.39 & \textbf{0.54}& \textbf{0.1}\\ 
MBPP (fine-tuning) & 11.31 & 0.40 & \textbf{0.51} & \textbf{0.1} & 11.64 & 0.39 & \textbf{0.51} & \textbf{0.1}\\ 
BigCodeBench & N/A & N/A & N/A & N/A & 16.65 & 0.47 & \textbf{0.32} & 0.284\\ 
\bottomrule
\end{tabular} 
}\\
\vspace{1pt}
\resizebox{0.92\textwidth}{!}{
\begin{tabular}{lcccccccccccccccc}
\toprule
\multirow{3}{*}{\textbf{CodeRA}} & \multicolumn{7}{c}{Training Data $D$} & \multicolumn{7}{c}{Testing Data $D_{test}$} \\ \cmidrule(rl){2-9} \cmidrule(rl){10-17}
& Bleu & Rouge & CodeBleu & Fusi &$ASR^{priv}_{x, D}$ & $ASR^{prop}_{x, D}$ & $ASR^{fusi}_{x, D}$ & $CRT_y$ & Bleu & Rouge & CodeBleu & Fusi & $ASR^{priv}_{x, D_{test}}$ & $ASR^{prop}_{x, D_{test}}$ & $ASR^{fusi}_{x, D_{test}}$ & $CRT_y$ \\ \midrule
MBPP (no fine-tuning) & 74.32 &0.86  &67.87 & 0.46 & \textbf{0.99} & \textbf{0.98} & \textbf{0.47} & \textbf{1.0} &73.00  &0.84  &66.42  & 0.44 & \textbf{0.99}  & \textbf{0.97} & \textbf{0.46} & \textbf{1.0} \\ 
MBPP (fine-tuning) & 72.41 & 0.85 & 64.42 & 0.45 & \textbf{0.99} & \textbf{0.97} & \textbf{0.46} & \textbf{1.0} & 71.11 &  0.83 &  63.17 & 0.43& \textbf{0.99} & \textbf{ 0.96} &\textbf{0.44} & \textbf{1.0} \\ 
BigCodeBench & N/A & N/A & N/A & N/A & N/A & N/A & N/A & N/A & 50.67 &  0.68 &  56.90 & 0.34& \textbf{0.99} & \textbf{ 0.99} &\textbf{0.35} & \textbf{1.0} \\ 
\bottomrule
\end{tabular} 
}
\label{ThreatModelValidation}
\end{table*}

\subsection{Code Reconstruction Attack}

In CodeRA, the cloud 1) feeds the prompt embeddings $\{\mathcal{E}\}$ to the cloud-hosted model $\theta_m$ concatenated with the original decoder $\theta_d$, denoted as $\theta_m \oplus \theta_d$, or 2) feeds the reconstructed prompts $\{\hat{x}\}$ to the LLM $\theta$, or 3) feeds the back-propagated gradients $\nabla_{x, y} \theta_m$ to BiSR \citep{chen2024unveilingvulnerabilityprivatefinetuning} to generates the output code $\hat{y}$ approximating the ground-truth code $y$, i.e., CodeRA: $(\{\mathcal{E}\}, \theta_m \oplus \theta_d) \cup (\{\hat{x}\}, \theta) \cup (\{\nabla_{x, y} \theta_m\}, \theta) \rightarrow \hat{y}$. The closer $\hat{y}$ is to $y$, the stronger the attack and the higher the risk are. This similarity is quantified by: \textbf{(1)} Information Privacy: token-level metrics (Bleu and Rouge) to measure similarity between the reconstructed and the ground-truth code ($\hat{y}$ and $y$); \textbf{(2)} Code Confidentiality: CodeBleu \citep{CodeBleu} for code syntax and semantic similarity between $\hat{y}$ and $y$; and \textbf{(3)} Code Functionality: similarity in unit tests passed by $\hat{y}$ and $y$, defined as follows:
\begin{small}
\begin{equation}
Fusi(\hat{y}, y) = \frac{\sum_{u \in U_x}\mathbb{I}[pass(\hat{y},u)=1 \text{ \& } pass(y,u) = 1]}{\sum_{u \in U_x}\mathbb{I}[pass(y,u) = 1]},
\label{functionality match}
\end{equation}
\end{small}
where $U_x$ is a set of unit tests $u$ to access the functionality of the generated code, given a prompt $x$. $pass(y, u) = 1$ if the code $y$ passes the unit test $u$; otherwise, $pass(y, u) = 0$. $\mathbb{I}[\cdot] = 1$ if $pass(\hat{y}, u) = pass(y, u) = 1$; otherwise $\mathbb{I}[\cdot] = 0$.

The cloud wins the game for a given code $y$ in the training set $D$ if it returns $\hat{y}$ with a clear gist of information privacy, code confidentiality, and code functionality. Fusi is stricter than Bleu, Rouge, and CodeBleu since passing unit tests demands full semantic and syntax correctness. A single token error can fail unit tests while barely affecting Bleu, Rouge, and CodeBleu scores. Hence, Fusi yields lower values than Bleu, Rouge, and CodeBleu, and lower ASRs for code functionality than for privacy and confidentiality. 
Passing even one similar unit test with the ground-truth code signals code functionality leak. Therefore, the cloud wins the functionality security game if $Fusi(\hat{y}, y) > \rho_{f} = 0$ (a leak vs. no leak).

The cloud wins the game for code $y$ in terms of i) information privacy if $Bleu(\hat{y}, y) \geq \rho_b$ $(=20)$ \citep{BleuScoreRange} or $Rouge(\hat{y}, y) \geq \rho_r$ $(=0.4)$ \citep{ROUGEScoreRange}, ii) code confidentiality if $CodeBleu(\hat{y}, y) \geq \rho_b$ $(=20)$, and iii) code functionality if $Fusi(\hat{y}, y) > \rho_f$ $(=0)$. The ASR is the cloud's average winning rate over all the code in $D$. During inference, the cloud attacks with one request for all prompt embeddings in the client's test set.

\subsection{Token-level Privacy Metrics}

Alongside Bleu, Rouge, CodeBleu, and Fusi metrics, we report correctly reconstructed tokens (CRT) in prompt $x$ and code $y$ defined as: $\bm{CRT_x = \frac{1}{|D|}\sum_{x \in D}\frac{|x \cap \hat{x}|}{|x|}}$ and $\bm{CRT_y = \frac{1}{|D|}\sum_{y \in D}\frac{|y \cap \hat{y}|}{|y|}}$.
We also assess the cloud ability to reconstruct sensitive information (imported packages, variable names, and function names) via $\textbf{\texttt{leak}} \bm{= \frac{1}{|D|}\sum_{y \in D}\textbf{\texttt{leak}}(y)}$, where $\texttt{leak}(y) = 1$ if such sensitive information appears in the reconstructed code, and $\texttt{leak}(y) = 0$ otherwise.
Together, these metrics comprehensively capture privacy leakage at both structural and token levels in PromptRA and CodeRA.

\subsection{Initial Assessment of PromptRA, CodeRA}

Our initial focus is to assess RAs through an experiment in a defense-free environment using benchmark CodeAlpaca \citep{codealpaca} and MBPP \citep{austin2021program} datasets from Evalplus, both designed for Python code generation. We use CodeLlama-7B as the LLM $\theta$, with the first attention block as the encoder $\theta_e$, the last four attention blocks as the decoder $\theta_d$, and the remaining middle blocks forming the cloud-hosted model $\theta_m$. CodeAlpaca contains $\sim$18k data points, which is larger than MBPP, which has 974 data points.
The cloud initializes the BaseModel and Vec2Text from a T5-base checkpoint \citep{raffel2023exploringlimitstransferlearning} as in \citep{morris-etal-2023-text}, and trains PromptRA on the (larger) CodeAlpaca as public data (batch size: 24, the max sequence length: $768$ tokens, since the original texts in CodeAlpaca are long). While, the client fine-tunes its encoder and decoder on the (smaller) MBPP.

Without fine-tuning and without test cases in client prompts, PromptRA achieves high Bleu scores (34.17 and 35.53 for training and test sets respectively, Table \ref{ThreatModelValidation}) and strong Rouge scores (0.66 and 0.67).
Consequently, the privacy attack success rates $ASR_{x, D}^{priv}$ and $ASR_{x, D_{test}}^{priv}$ reach 96\%, showing that the cloud can reconstruct meaningful content in most client prompts. 
When test cases are included in client prompts\footnote{For instance, given the prompt ``Write a Python function to find the remainder of two numbers,'' a test case is `assert find(3,3)==0.' \label{test-cases}}, Rouge scores drops: 0.4 and 0.39; resulting in lower ASRs: 50\% and 54\%, for training and test sets respectively. This is because Vec2Text \citep{morris-etal-2023-text} and BiSR \citep{chen2024unveilingvulnerabilityprivatefinetuning} are not suited for long prompts with code and tests. Although PromptRA weakens, CodeRA exels with Bleu 74.32, Rouge 0.86, CodeBleu 67.87, and Fusi 0.46, when prompts have test cases, yielding high ASRs: over 97\% privacy, 98.25\% code confidentiality, and 46.5\% code functionality leakage across training and test sets. 

The cloud still registers \textbf{high ASRs} in PromptRA and CodeRA \textbf{when the client fine-tunes and conceals its encoder and decoder}. Privacy ASRs reach 51\% for prompts, 95.5\% for prompts without test cases, and 99\% for code. Code confidentiality ASR is 97\%, and code functionality ASR is 45.75\% across training and test sets on average.
Similar results appear on HumanEval \citep{chen2021evaluating}, BigCodeBench \citep{zhuo2024bigcodebench}, and LLMs such as CodeQwen1.5-7B-Chat, Llama3-8B-instruct. We observe token-level metrics $CRT_x = 0.28$ and $CRT_y = 1.0$ on MBPP. The same attack performance is registered on the BigCodeBench, an advanced evaluation of LLMs in programming with 1,140 data points. The \texttt{leak} scores of sensitive information are high (0.98) on both MBPP and BigCodeBench.

\textbf{Remarks.} A cloud can reconstruct the content of prompts and code with high Bleu, Rouge, CodeBleu, Fusi, CRT, and \texttt{leak} scores, leading to high ASRs. Even one high ASR poses significant IP and data security risks for the client. Thus, effective privacy-preserving mechanisms are essential to protect prompts and code while sustaining model performance.

\section{IND-preserving Vocabulary}
\label{LDP-preserving Vocabulary}

RAs infer tokens in the prompt $x$ from its embedding $\mathcal{E}$ to reconstruct prompts and code.
At the token level, $x = \{t_j\}_{j =1}^{|x|}$ is represented as a sequence of token embeddings $\{e_{t_j}\}_{j = 1}^{|x|}$, where $t_j \in V$. The closer reconstructed token embeddings are to the ground-truth ones in $x$, the more accurately the cloud can infer every token $t_j$, and thus reconstruct $x$ and its code $y$.
To provide privacy guarantees, we consider the worst-case attack where the cloud losslessly reconstructs the ground-truth token embeddings: 
$\forall t_j \in x: \hat{e}_{t_j} = e_{t_j}$, with $\hat{e}_{t_j}$ the reconstructed token embedding of the token $t_j$.

To defend against this worst-case, the client can randomize tokens in the prompt $x$ and the code $y$ for local fine-tuning, resulting in replacing tokens with other tokens to preserve $d_x$-privacy \citep{10.1007/978-3-642-39077-7_5}. The client can then either denoise the cloud-hosted LLM output \citep{mai2024splitanddenoise} or fine-tune the decoder for private code generation as in \textsc{NOIR}. Traditionally, $d_x$-privacy at the token level aims to prevent a data curator from authorship inference \citep{mattern-etal-2022-limits}, not reconstruction attacks as in our study. This mismatch leads to the following fundamental challenges.

\textbf{Challenges.} A subtle token-level perturbation in the prompt $x$ can drastically alter its content (syntax, functionality, function names, variables, etc.), making LLMs generate irrelevant or faulty code. When a token $t$ appears in multiple prompts, perturbing its token embedding in these prompts with different draws of random noise to achieve $d_x$-privacy amplifies randomness, as $t$ is represented by inconsistent randomized token embeddings, degrading model performance. Similarly, the randomness is applied to the instruction $\pi$ and the template $\mathcal{T}$ weakening LLM performance, since poor and noisy instructions further reduce the output quality. Thus, balancing strong performance in code generation with LDP against PromptRA and CodeRA remains challenging.



\subsection{\textsc{INDVocab}}

To address the challenge, we propose a novel concept of indis-tinguishability-preserving vocabulary (\textsc{INDVocab}), which is a LDP protection at the token embedding level, as follows: 

\begin{mydef}{$\epsilon$-\textsc{INDVocab}.}
A randomized algorithm $\mathcal{M}$ fulfills $\epsilon$-INDVocab, if for any two tokens $t$ and $t'$ in the vocabulary $V$, and for all possible outputs $\mathcal{O}$ of $\mathcal{M}$, i.e., $\mathcal{O} \in Range(\mathcal{M})$, we have: $Pr[\mathcal{M}(e_t) = \mathcal{O}] \leq e^{\epsilon} Pr[\mathcal{M}(e_{t'}) = \mathcal{O}]$.
\label{LDPVocab}
\end{mydef}

$\epsilon$-\textsc{INDVocab} ensures that the cloud cannot distinguish the outcomes $\mathcal{M}(e_t)$ of the original token embeddings $e_t$ under $\epsilon$-indistinguishability protection.
 Instead of using the original token embeddings $\{e_t\}_{t \in V}$ in the vocabulary $V$, the client applies 
algorithm $\mathcal{M}$ to randomize these token embeddings resulting in IND-preserving token embeddings $\{\tilde{e}_t = \mathcal{M}(e_t)\}_{t \in V}$ in a new \textsc{INDVocab} $\tilde{V}$. By locally replacing $V$ with $\tilde{V}$ to map a token $t$ to an embedding $\tilde{e}_t$, the client prevents the cloud from inferring the ground truth tokens. This is because the cloud can only reconstruct the randomized token embeddings $\{\tilde{e}_t\}_{t \in V}$ instead of the ground-truth token embeddings. Thus, the tokens and the prompt $x$ are protected against PromptRA.
Using the \textsc{INDVocab} $\tilde{V}$ to derive the prompt embedding by feeding the prompt $x$ to the encoder $\theta_e$ results in an IND-preserving prompt embedding $\tilde{e}$ following the post-processing property in DP \citep{dwork2014}. The IND-preserving tokens, prompts, and prompt embeddings $\tilde{e}$ further prevent the cloud from reconstructing the code $y$ with CodeRA.


\textbf{Achieving \textsc{INDVocab} with Adaptive Randomized Response (ARR).} We propose an ARR mechanism as the algorithm $\mathcal{M}$ in Def. \ref{LDPVocab} to preserve \textsc{INDVocab} while maintaining high model utility with key advantages to overcome the challenges and limitations related to token-level LDP \citep{mattern-etal-2022-limits} in code generation. The pseudo-code of our ARR mechanism is in Alg. \ref{CodeX - Psuedo Code}, Lines 4-8. Let us denote the $i^\text{th}$ feature in a token embedding $e_t$ as $e^i_t$. The ARR's key idea is flipping a probabilistic coin whether we keep the original value of the feature $e^i_t$ or change it to another possible feature value among other tokens $\{e^i_k\}_{k \in V \setminus t}$, where $k$ is a token different from $t$ in the vocabulary $V$, such that \textit{``more similar feature values to $e^i_t$ have higher probabilities to be selected as a replacement.''} We formulate the idea as follows:
\begin{align}
& \text{Given an arbitrary token } t \in V, \forall i^\text{th}\text{-feature of } e_t: \label{ARR} \\
& \tilde{e}^i_t = \left \{\begin{small}
  \begin{aligned}
    & e^i_t, \text{ with probability } p_i = \exp({\beta_i}) / [\exp({\beta_i}) + |V| -1], \\
    & e^i_k, \text{ with probability } q_{i,k} =  (|V| -1)q_k / [\exp({\beta_i})  + |V| -1],
  \end{aligned}\end{small} \right. 
  \nonumber
\end{align}
\noindent where $k \in V \setminus t$, $q_k = \exp(-\Delta^i_{t,k}/m) / \sum_{l \in V} \exp(-\Delta^i_{t,l}/m)$, $\Delta^i_{t, k} = |e^i_t - e^i_k|$, $m$ is the number of features in a token embedding, and $p_i + \sum_{k \in V \setminus t} q_{i,k} = 1$. The larger $q_k$ indicates that $e^i_k$ is more similar to $e^i_t$. Theorem \ref{theorem-beta-bound} bounds $\beta_i$ s.t. $\forall k \in V \setminus t: p_i \geq q_{i,k}$ preserving $\epsilon_i$-IND as in typical RR definition \cite{warner1965randomized}.

\begin{theorem} Randomizing the $i^\text{th}$-feature in a token embedding $e_t$ preserves $\epsilon_i$-IND if $\varepsilon_i$ and $\beta_i$ are bounded as
\begin{footnotesize}
\begin{align} 
& \forall \varepsilon_i \geq \frac{1}{m}(\Delta^i_{t, max} - \Delta^i_{t, min}): \ln (|V| - 1) + \ln (\frac{\max\{\exp(-\Delta^i_{t,k}/m)\}_{k \in V \setminus t}}{ \sum_{l \in V} \exp(-\Delta^i_{t,l}/m)}) \nonumber \\
& \leq \beta_i \le \varepsilon_i + \ln(|V|-1) + \ln (\frac{\min \{\exp(-\Delta^i_{t,k}/m)\}_{k \in V \setminus t}}{ \sum_{j \in V} \exp(-\Delta^i_{t,j}/m)}), \label{betabound}
\end{align} \end{footnotesize}
\noindent where $\Delta^i_{t, min} = \min \{\Delta^i_{t, k}\}_{k \in V \setminus t}$, $\Delta^i_{t, max} = \max \{\Delta^i_{t, k}\}_{k \in V \setminus t}$.
\label{theorem-beta-bound}
\end{theorem}

The proofs of Theorem \ref{theorem-beta-bound} and the following theorems are in the Appx.
In practice, we use the upper-bound of $\beta_i$ in Theorem \ref{theorem-beta-bound}, since larger $\beta_i$ yields less noisy randomization probabilities. We apply ARR to independently randomize every $i^{\text{th}}$-feature in a token embedding $e_t$ with a privacy budget $\epsilon_i$. Thus, we create an $\epsilon$-IND-preserving token embedding $\tilde{e}_t$ where the total privacy budget is $\epsilon = \sum_{i \in e_t}\epsilon_i$ following the composition theorem in DP \citep{dwork2014} reflected in Theorem \ref{Theorem-composition} (Appx. B). The total IND budget $\epsilon$ is lower-bounded as follows: $\epsilon = \sum_{i \in e_t}\epsilon_i \geq \frac{1}{m}\sum_{i \in e_t}\Delta_{t, max}^i$ (Eq. \ref{betabound}), which is tiny ($\lesssim 1e-3$) in almost all LLMs, enabling us to work with a high privacy protection (if needed) in practice.

\textbf{Post-processing Property.} The client applies \textsc{INDVocab} $\tilde{V}$ to represent every prompt $x$ and the instruction $\pi$ as a sequence of $\epsilon$-IND-preserving token embeddings, e.g., the sequence of token embeddings given a prompt $x$ is $\{\tilde{e}_{t_j}\}_{t_j = 1}^{|x|}$. The number of times a token $t$ appears in one or more prompts does not affect the $\epsilon$-IND guarantee of the token embedding of $t$ following the post-processing property of DP \citep{dwork2014}; that is, no extra information regarding the ground-truth token's embedding is used afterward. The client freely uses the \textsc{INDVocab} $\tilde{V}$ without affecting the IND protection.

\subsection{Prompt-Level Protection}


In this study, we safeguard prompts and code against reconstruction attacks from the cloud. Thus, we provide an \textbf{upper-bounded probability of reconstructing a given prompt $x$} in the security game described in Section \ref{Prompt Reconstruction Attack}, highlighting the link between the token-level \textsc{INDVocab} and prompt-level privacy under reconstruction risk. In this game, reconstructed and ground-truth prompts $\hat{x}$ and $x$ have the same number of tokens: $|\hat{x}| = |x|$. Hence, $Rouge\text{-}F1(\hat{x}, x) = Bleu(\hat{x}, x) = C / |x|$, where $C$ is the number of correctly reconstructed tokens in $\hat{x}$. Our analysis in Theorem \ref{upper-bound token} and Proposition \ref{Bleu Bound} (Appx. \ref{Mitigating Limitations of Token-level Privacy}) upper-bounds the cloud's probability of recovering ground-truth tokens in $x$ with gist level higher than or equal to $\rho$: 
\begin{equation}
\small
P\big[C/|x| \geq \rho; \{o_j\}_{j = 1}^{|x|}\big] \leq \big(\frac{\psi e^\epsilon + 1}{\psi e^\epsilon + \psi^2}\big)^{\rho |x|} \times \big(\frac{\psi e^\epsilon}{\psi e^\epsilon + 1}\big)^{(1-\rho)|x|},
\label{bound1}
\end{equation}
where $\psi = |V| - 1$. We generalize Proposition \ref{Bleu Bound} to consider the reconstruction attack with token sequences, capturing correlation among tokens in real-world scenarios, as follows: The cloud's previously reconstructed token sequences $\hat{t}_{<j}$ can enhance its probability of correctly reconstructing the next token $t_j$: $Pr(\hat{t}_j = t_j | \hat{t}_{<j})$. This advantage is bounded by a constant $\gamma \in [0, 1]$ in practice, as follows: 
$\forall t_j \in x: 0 \leq Pr(\hat{t}_j = t_j | \hat{t}_{<j}) - Pr(\hat{t}_j = t_j) \le \gamma.$
\indent Given $\gamma$, Theorem \ref{Token sequences bound} (Appx. \ref{Mitigating Limitations of Token-level Privacy}) tightens the upper-bounded reconstruction risk, as follows:
\begin{footnotesize}
\begin{equation}
P\big[C/|x| \geq \rho; \{o_j\}_{j = 1}^{|x|}\big] \leq \big(\frac{\psi e^\epsilon + 1}{\psi e^\epsilon + \psi^2}+\gamma\big)^{\rho |x|} \times \big(\frac{\psi e^\epsilon}{\psi e^\epsilon + 1}-\gamma\big)^{(1-\rho)|x|}.
\label{bound2}
\end{equation}
\end{footnotesize}
\indent \textbf{Meaningful Level of Privacy Protection.} We analyze this upper-bounded reconstruction risk as a function of $\epsilon$, prompt size $|x|$, vocabulary size $|V|$, and reconstruction threshold $\rho$, showing the effectiveness of \textsc{INDVocab} in protecting the prompt $x$. With $\epsilon = 13$, the (clear gist) upper-bounded reconstruction risk for $|x|=200$ is $<5.5 \times 10^{-11}$ (Eq. \ref{bound2}) for modern LLMs (151k+ tokens in their vocabulary) with the small upper-bounded constant\footnote{We compute the upper-bounded $\gamma$ with $Pr(\hat{t}_j = t_j | \hat{t}_{<j})$ for every prompt in the dataset $D$: $\bm{\gamma=\max_{j \in [0, |x|-1]} \Sigma_{x\in D}\mathbb{I}(\hat t_{j} = t_{j})/|D|}$.
} $\gamma \leq 0.146$ derived from Evalplus datasets, which shrinks the reconstruction probability (Eq. \ref{bound1}). The upper-bounded reconstruction risk is equivalent to guessing an 8-character lowercase password ($1/26^8$). Unlike password attacks (repeated guesses), NOIR allows only one guess per prompt, as clients update their inputs frequently, making brute-force success impractical. Even with constant submissions of the same prompt every second, success would take an estimated $\frac{1}{2} \times \frac{26^8}{31,557,600} = 3,308.7$ years (one guess/s), in which $31,557,600$ is the number of seconds/year. Smaller $\epsilon$ (e.g., 9) exponentially enhance security, offering $\frac{1}{2} \times \frac{26^{72}}{31,557,600} = 2.4e^{94}$ years and impenetrable protection for shorter prompts (20 words) against brute-force attackers. Longer prompts have lower upper-bounded reconstruction probabilities. This $\epsilon$ range with small upper bounds for modern LLMs and typical prompts maintains high model performance while providing robust protection against RAs.

\subsection{Advantages of \textsc{INDVocab}}
\label{Reconstruction Upper Bounds}

\textsc{INDVocab} has fundamental advantages to achieving good model utility while protecting prompts and code. 

\textbf{Unchanged Tokens and Negligible Noise.} 
If tokens in prompts $x$, instruction $\pi$, and template $\mathcal{T}$ are replaced by other tokens as in T2T \citep{10.1145/3459637.3482281,10.1145/3336191.3371856,li2023privacypreserving}, the code syntax breaks since an inappropriate tokens degrade functionality and distorts the essential correlation between a prompt $x$ and its output code $y$. Fine-tuning a model on such distorted pairs yields models that generate code misaligned with prompt requirements. Even without changing tokens in prompts, the instruction, and the template, randomizing every token embedding with different draws of considerable noise as in SnD \citep{mai2024splitanddenoise} still distorts the correlation among $x$, $\pi$, $\mathcal{T}$, and code $y$ since multiple token embeddings can represent one token. As a result, the model performance is remarkably degenerated.

\begin{figure}[t]
\captionsetup[subfigure]{justification=centering}
    \centering
      \begin{subfigure}[t]{0.235\textwidth}
        \centering
       \includegraphics[scale=0.35]{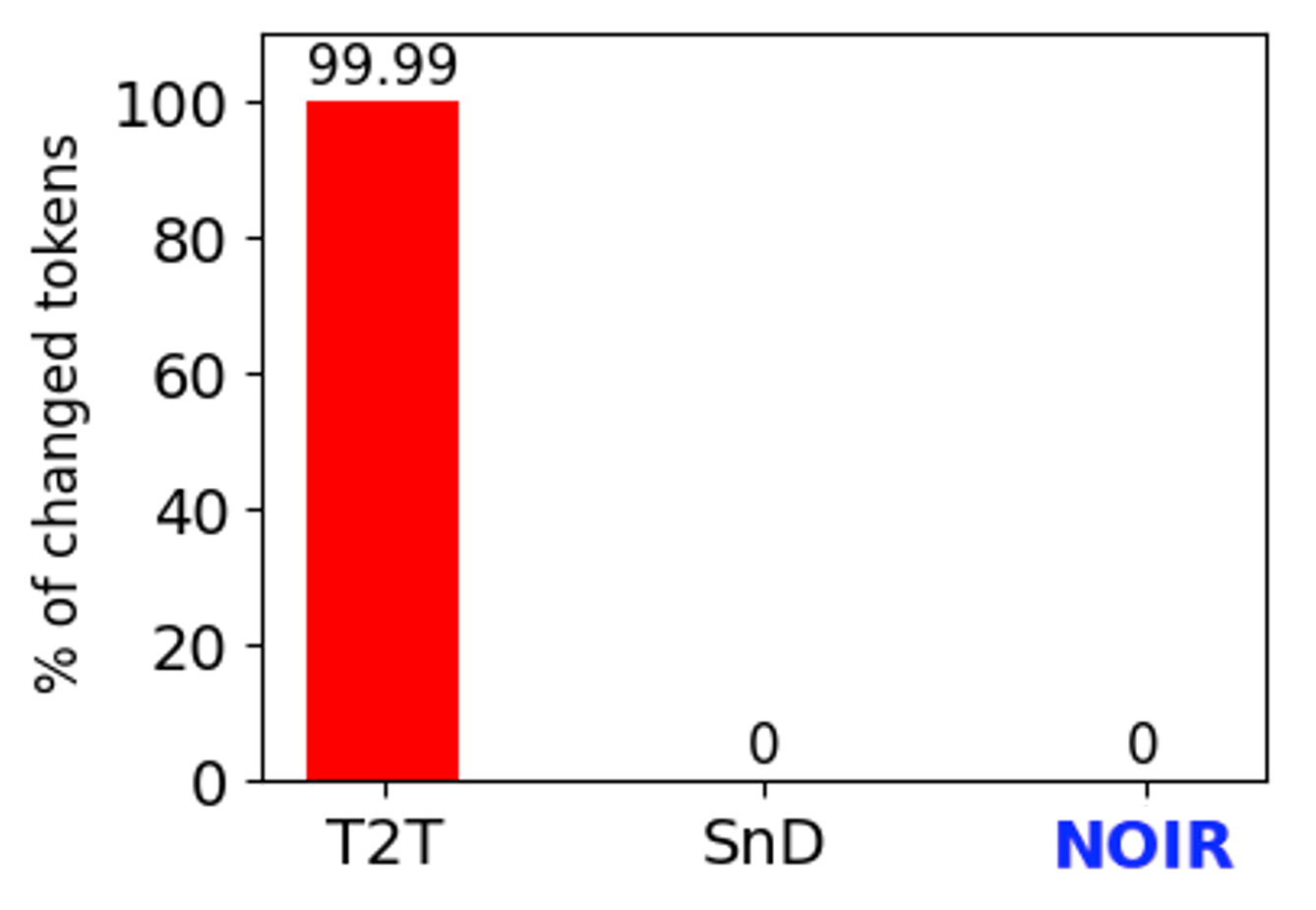} \vspace{-4pt}
       \caption{} \vspace{-4pt}
       \label{Fig::PercentTokenChange}
       \end{subfigure}
       \hfill
      \begin{subfigure}[t]{0.235\textwidth}
      \centering
      \includegraphics[scale=0.19]{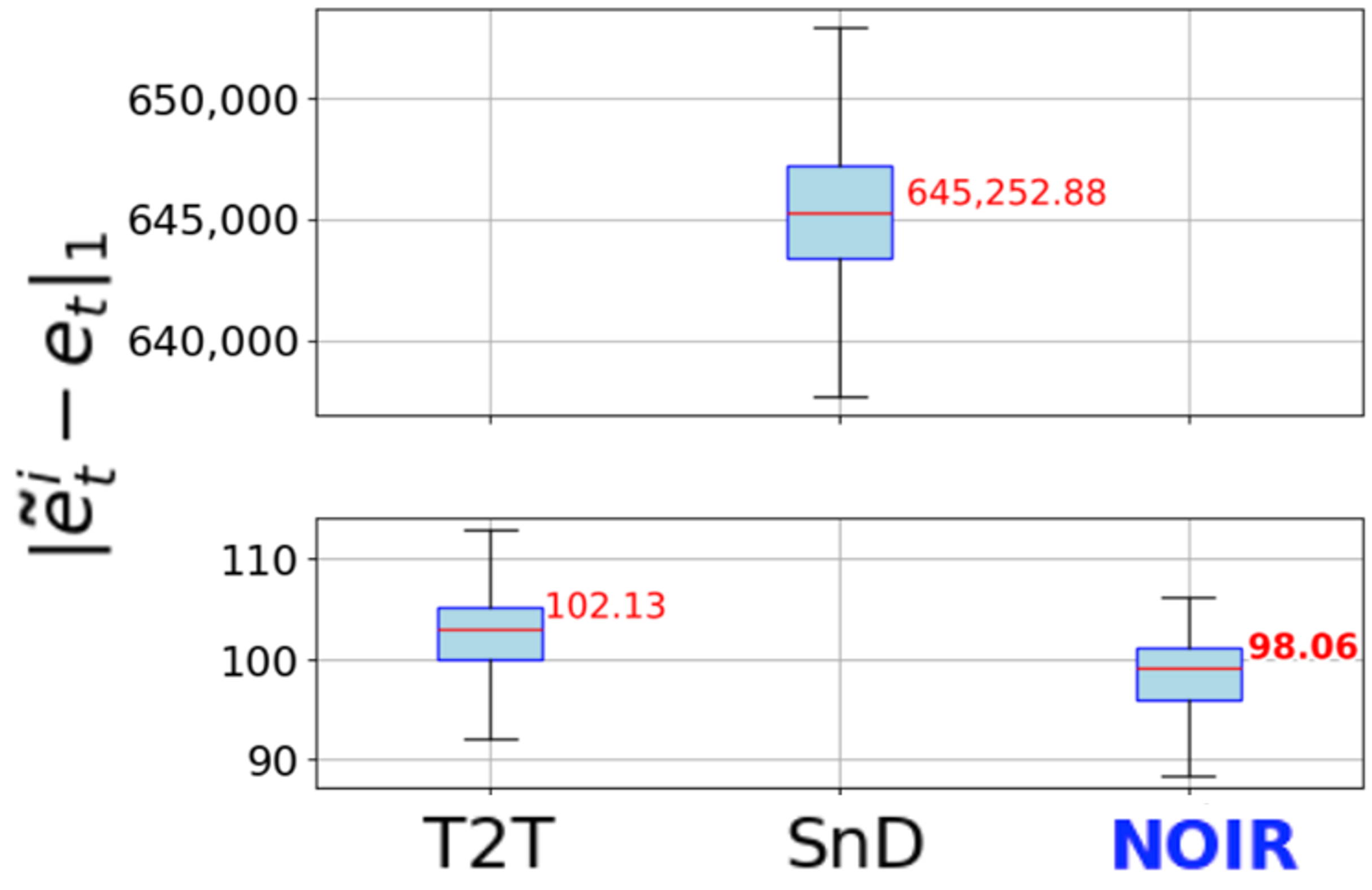} \vspace{-4pt}
        \caption{} \vspace{-4pt}
        \label{Fig:EmbDiff}
    \end{subfigure}
    
    \caption{(a) the percentage of tokens changes in a prompt, (b) the $L_1$-norm distance between original and IND-preserving token embeddings on the MBPP dataset, $\epsilon = 13$.}
    \label{Fig:TokenChangeAndEmbDiff}  
\end{figure}

\textsc{INDVocab} preserves all the tokens in the prompts $x$, instruction $\pi$, and the template $\mathcal{T}$, randomzing tokens' embeddings once with negligible $\epsilon$-IND noise.
This property retains the essential correlations among $x$, $\pi$, and $\mathcal{T}$ in the client's local data, enabling fine-tuning of the encoder and decoder while balancing utility and privacy. We evaluate this with four experiments using the MBPP dataset: \textbf{(1)} Average percentage of changed tokens in prompts with test cases after IND preservation: $\frac{1}{|D|}\sum_{x \in D} \frac{\#\text{Changed Tokens in } x}{|x|}$; \textbf{(2)} Average $L_1$-norm distance between original and IND-preserved token embeddings using 32k tokens of the CodeLlama-7B model's vocabulary: $\{|\tilde{e}_t - e_t|_1\}_{t \in x, x \in D}$; \textbf{(3)} Change in (angle) cosine similarity between bi-gram (two sequential) tokens $t_i$ and $t_{i+1}$ in prompts $x$: $\{|cos(e_{t_i}, e_{t_{i+1}}) - cos(\tilde{e}_{t_i}, \tilde{e}_{t_{i+1}})|_1\}_{t_i \in x, x \in D}$, where $cos(\cdot)$ is a cosine similarity function; and \textbf{(4)} Change in (angle) cosine similarity between any token pairs in \textsc{INDVocab} given the original vocabulary: $\{|cos(e_{t}, e_{t'}) - cos(\tilde{e}_{t}, \tilde{e}_{t'})|_1\}_{t,t' \in V, t \neq t'}$.
For a fair comparison, the privacy budget $\eta d_x$ in T2T and SnD ($\eta$ is predefined) equals the $\epsilon$ in \textsc{NOIR}.

Figure \ref{Fig::PercentTokenChange} shows that T2T \citep{10.1145/3459637.3482281,10.1145/3336191.3371856,li2023privacypreserving} alters nearly all tokens ($\approx 100\%$) in every prompt and test case, whereas SnD \citep{mai2024splitanddenoise} and \textsc{NOIR} leave tokens unchanged. Despite this, SnD injects substantial noise into token embeddings, with an average $L_1$-norm distance of $\sim$645k versus 98.06 for \textsc{NOIR} (Figure \ref{Fig:EmbDiff}). The average $L_1$-norm distance in \textsc{NOIR} (102.13) is also significantly smaller than T2T, statistically significant: $p$-value $<3.3e$-$100$ (2-tail t-test). Thus, only \textsc{NOIR} preserves tokens while injecting negligible $\epsilon$-IND noise into their embeddings.

\begin{figure}[t]
\captionsetup[subfigure]{justification=centering}
    \centering
      \begin{subfigure}[t]{0.23\textwidth}
        \centering
       \includegraphics[scale=0.25]{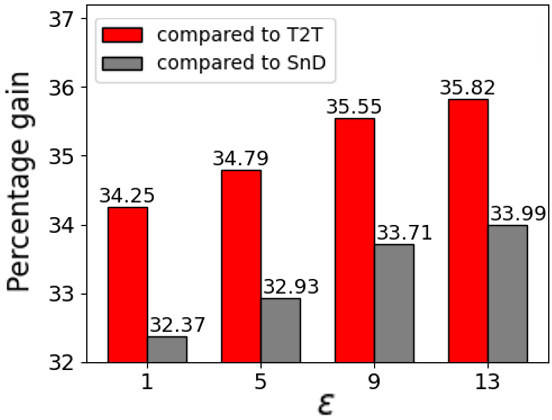}
       \caption{Gain in bi-gram angular}
    \label{Fig::analysisPercentageGain}
       \end{subfigure}
       \hfill
      \begin{subfigure}[t]{0.23\textwidth}
      \centering
      \includegraphics[scale=0.285]{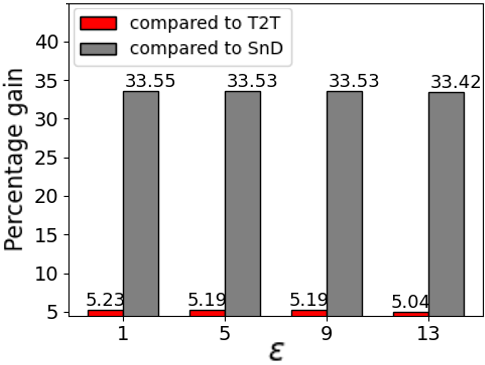}
        \caption{Gain in token angular}
        \label{Fig:AnlyticalResult}
    \end{subfigure}
    
    \caption{\textsc{NOIR}'s gain in terms of (a) bi-gram angular change in every prompt and (b) angular change among tokens in the vocabulary compared with T2T and SnD.}
    \label{Fig:AngularAnalysis}  
\end{figure}

\textsc{INDVocab} obtains a significantly smaller average bi-gram angular change than T2T and SnD across IND budgets $\epsilon$, with over 33\% improvement (Figure \ref{Fig::analysisPercentageGain}). Larger $\epsilon$ values further enhance \textsc{INDVocab}'s ability to preserve bi-gram angles compared to SoTA methods, thereby better maintaining angular correlations under IND protection. Figure \ref{Fig:AnlyticalResult} also shows that \textsc{NOIR} achieves the smallest average angular change, confirming \textsc{INDVocab}'s superior ability to preserve the relative angular correlation among tokens. 

\begin{wrapfigure}{r}{0.55\columnwidth}
  \begin{center} \vspace{-10pt}
\includegraphics[width=0.5\columnwidth]{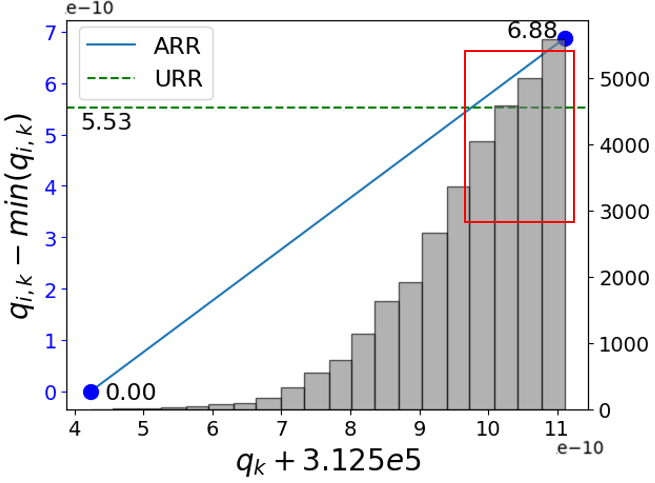} \vspace{-10pt}
  \end{center}
  \caption{Adaptive probabilities.} \vspace{-10pt}
\label{Fig::AdaptiveProb}
\end{wrapfigure}
\textbf{Adaptive Randomization.}
To assess adaptive randomization, we compare the ARR with uniform randomized response (\textbf{URR}), where every token embedding feature has the same randomization probability.
In this experiment, we randomly pick one token $t$ and a feature $i^{\text{th}}$ in its token embedding. 
Figure \ref{Fig::AdaptiveProb} shows the correlation between the probability $q_{i, k}$ and the feature similarity $q_k$ under ARR and URR. Features similar to $e^i_t$ are more likely to be selected in the ARR than under URR's uniform probability. Those with similarity $q_k$ and probabilities $q_{i, k}$ higher than the ones in URR are the most frequent (red rectangle). Thus, ARR preserves embedding utility better than URR under the same IND protection.

\textbf{Agnostic to Embedding Size.} We evaluate various existing LDP mechanisms, from classical ones such as Gaussian \cite{dwork2014}, Laplace \cite{dwork2014}, and Duchi’s mechanisms \cite{duchi2013local,duchi2019lower}, to advanced methods including Piecewise \cite{wang2019collecting}, Hybrid \cite{wang2019collecting}, Three-outputs \cite{zhao2020local}, Adaptive OME \cite{lyu2020towards}, and XRand \cite{Nguyen_Lai_Phan_Thai_2023}. All face two key challenges in our context:
\textbf{(1)} High sensitivity $\Delta$ introduces excessive noise due to large embedding magnitudes;
\textbf{(2)} Privacy budget accumulation over large embedding size (e.g., $m = 4,096$) renders generation quality unusable—outputs resemble code at $\epsilon \approx 500,000$. Even relaxed $d_x$-privacy-based approaches \citep{10.1145/3459637.3482281,10.1145/3336191.3371856,li2023privacypreserving,mai2024splitanddenoise} yield code-like outputs at $\epsilon \approx 50,000$.
Our ARR mechanism resolves these issues by advancing randomization probabilities  $q_{i, k}$ and $p_i$ using $-\Delta^i_{t,k}/m$ (Eq. \ref{ARR}), injecting negligible noise as demonstrated above. The total privacy budget becomes agnostic to embedding size $m$: $\epsilon = \sum_{i \in e_t}\epsilon_i \geq \frac{1}{m}\sum_{i \in e_t}\Delta_{t, max}^i$ (Theorem \ref{theorem-beta-bound}).

\textbf{Remark.} Our \textbf{ARR} mechanism \textbf{is uniquely suitable for LLM generation} under strict IND guarantees.

\section{\textsc{STuning} and \textsc{LTokenizer}}
\label{STuning}


\textbf{Model Misalignment.} After deriving an effective \textsc{INDVocab} to prevent the cloud from reconstructing client prompts and code, we now focus on fine-tuning the encoder $\theta_e$ and decoder $\theta_d$ so the client can privately generate desirable code in \textsc{NOIR}. Despite the benefits of \textsc{INDVocab}, fine-tuning remains challenging: its privacy protection causes a misalignment between the encoder's output distribution and the cloud-hosted LLM's input distribution. This misalignment, amplified through the latent space of $\theta_m$, further disrupts the alignment between encoder output and decoder input, degrading performance as the trade-off for IND protection.



\textbf{\textsc{STuning}.} To address the model misalignment, we use split tuning: the client computes the loss function from the final-layer logits and the one-hot vector $v_t$ of the next token $t$ in the output $y$, derives the gradient of the decoder $\theta_d$, and sends the gradient up to the last layer of the middle block $\theta_m$ back to the cloud. The cloud continues the back-propagation computing the gradient of the middle block $\theta_m$'s LoRA and sends the gradient up to the client's encoder, where the client computes the gradient of the encoder $\theta_e$ locally. The client updates $\theta_e$ and $\theta_d$ while the cloud updates $\theta_m$'s LoRA. This fine-tuning minimizes the average loss: $\arg\min_{\theta} \frac{1}{|D|} \sum_{\{x, y\} \in D} L(\{\tilde{e}_{t_j}\}_{j = 1}^{|x|},y, \theta)$, where $|D|$ is the number of data points $\{x, y\}$ in $D$ and $\theta = \{\theta_e, \theta_d, \theta_m\text{'s LoRA}\}$, by updating $\theta$, as follows: 
$\theta \leftarrow \theta - \gamma \omega  \text{\ \ s.t. \ \ } \omega = \frac{1}{|D|} \sum_{\{x, y\} \in D}\partial L(\{\tilde{e}_{t_j}\}_{j = 1}^{|x|},y, \theta) / \partial \theta,$
where $\omega$ is the aggregated gradient and $\gamma$ is a learning rate. \textsc{STuning} computes gradients per data point $\partial L(\{\tilde{e}_{t_j}\}_{j = 1}^{|x|},y, \theta) / \partial \theta$ rather than over the aggregated statistics of multiple prompts, which damages the prompts' embeddings. \textsc{STuning} updates the model parameters using the aggregated gradients, recovering the impact of the $\epsilon$-IND noise until the model converges, yielding effective code generation models.

The power of \textsc{STuning} lies in its effectiveness and efficiency in model performance and fine-tuning cost. It requires no extra cloud information, preserving CPC constraints, and aligns the encoder and the decoder to mitigate the \textsc{INDVocab}'s impact. The cloud updates a lightweight LoRA at a negligible cost, while the client can fine-tune $\theta_e$ and $\theta_d$ directly—only one and four attention blocks are needed for encoder and decoder, respectively. As a result, \textsc{NOIR} delivers strong code generation performance cost-effectively, ideal for resource-limited clients. The cloud may skip fine-tuning $\theta_m$’s LoRA, lowering complexity with minimal utility loss. As shown in our \textbf{complexity and cost analysis} (Appx. \ref{Complexity and Cost Analysis}), \textbf{a single GPU suffices for client fine-tuning}, achieving high utility in code generation and completion.

\begin{figure}[t]
\centering
\resizebox{0.75\columnwidth}{!}{
\includegraphics{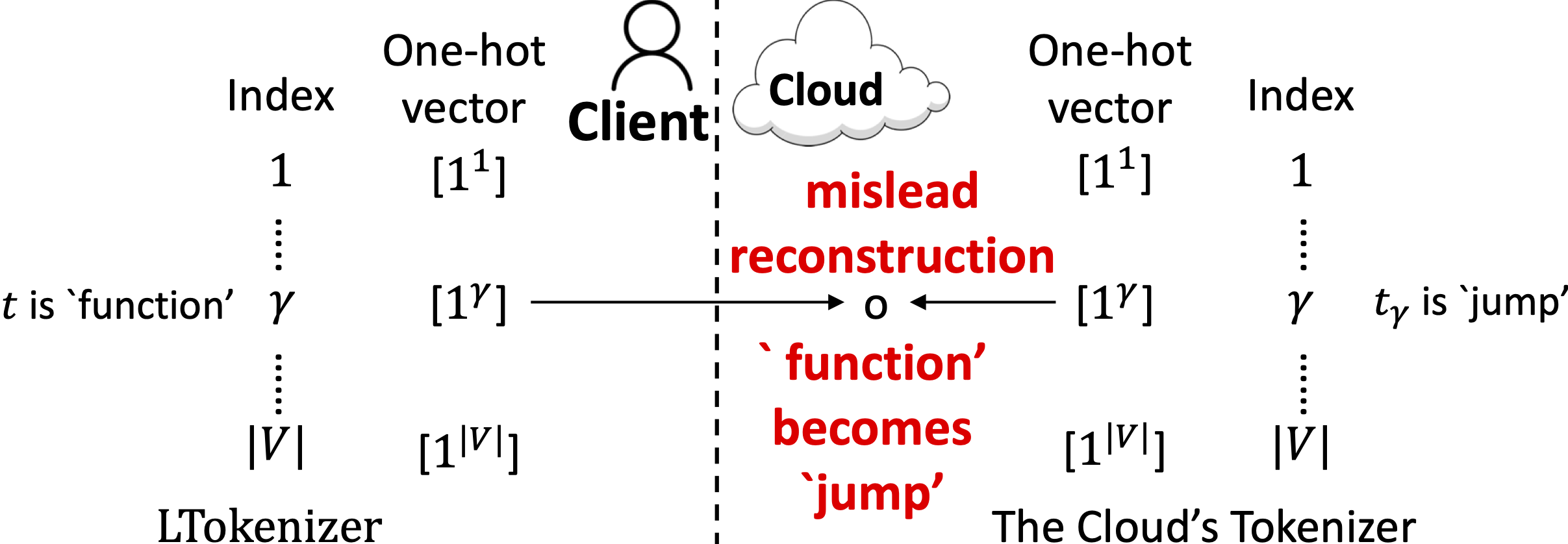}
}
\caption{Local Randomized Tokenizer (\textsc{LTokenizer}).}
\label{LTokenizer}
\end{figure}

\textbf{\textsc{LTokenizer}.} Although \textsc{STuning} is effective, the cloud can exploit the back-propagated gradients of the middle block $\theta_m$'s LoRA, derived from raw one-hot vectors $\{v_t\}_{t \in y}$ of output tokens $t$, to reconstruct training data through BiSR \citep{chen2024unveilingvulnerabilityprivatefinetuning}. Using the shared open-source tokenizer and the vocabulary $V$, the cloud maps each $v_t$ to the client's token $t$.
Sine \textsc{INDVocab} does not protect these vectors $\{v_t\}_{t \in y}$, \textsc{NOIR} introduces a local tokenizer (\textsc{LTokenizer}), which uniformly assigns every token and its embedding in the \textsc{INDVocab} $\tilde{V}$ to a random index in the tokenizer on the client side. 

Figure \ref{LTokenizer} illustrates assigning a token $t$ (e.g., \textit{`function'}) and its embedding $\tilde{e}_{t_1}$ in the \textsc{INDVocab} to a random index $\gamma$. The one-hot vector $v_{t}$ becomes $[1^\gamma]$, where  only the $\gamma^\text{th}$ element is 1 and all remaining elements are 0. 
If $t$ is the next token to be generated in the output $y$, the gradient minimizes the loss between $v_{t} = [1^\gamma]$ and the decoder's output logits. This gradient misleads the cloud's RAs into reconstructing the token at index $\gamma$ of the cloud's tokenizer with its vocabulary $V$, yielding an irrelevant token (e.g., \textit{`jump'}) instead of the client's token (\textit{`function'}).
Figure \ref{fig:running example} shows such a meaningless reconstruction. \textsc{LTokenizer} is data-independent and incurs no extra privacy cost.

The cloud's probability of inferring a client's token index is $1/|V|$, which is negligible; thus, the client's secret one-hot vectors $v_t$ and tokenizer remain protected. The client employs \textsc{LTokenizer} with \textsc{INDVocab} in both \textsc{STuning} and inference (Alg. \ref{CodeX - Psuedo Code}).

\textbf{Vulnerability to Averaging Attacks.} There is no need to randomize the vocabulary or one-hot vector assignments per prompt. PromptRA, CodeRA, and frequency analysis attacks are types of averaging attacks using public data to reconstruct tokens. Our evaluation shows they are ineffective against \textsc{NOIR}. As long as the cloud does not observe any auxiliary information about the client tokens, embeddings, prompts, or code beyond what \textsc{INDVocab} and \textsc{LTokenizer} protect, reconstruction risks remain upper-bounded. It is because averaging over IND-preserving token embeddings does not leak extra information due to DP's post-processing property. Averaging attacks might succeed only if an internal actor colludes with the cloud, which lies outside \textsc{NOIR}’s threat model.

\section{Experimental Results}
\label{Experimental Results}

Our extensive experiments shed light on the trade-off between client-side privacy and model utility against cloud RAs. \textsc{NOIR} surpasses the SoTA with strong privacy against the cloud while remaining cost-effective through minimal client fine-tuning, advancing protection of client prompts and code in generation and completion.

\textbf{Configurations.} We use CodeLlama-7B, CodeQwen1.5-7B-Chat, and Llama3-8B-instruct models with vocab sizes of 32k, 92,416, and 128,256 tokens, respectively; CodeLlama-7B is the default. The client runs one and four attention blocks per model as the encoder and the decoder, with the remaining blocks cloud-hosted. Each token embedding has $m$=$4,096$ features, with uniform $\epsilon/m$ across all features in $e_t$. The temperature $h$ is $0.25$, yielding the best performance across mechanisms. As aforementioned, we open-source \textsc{NOIR} as a privacy-preserving coding agent based on Qwen2.5-Coder-32B-Instruct via a web service and a VS extension.

\textbf{Datasets.} We evaluate \textsc{NOIR} using the Evalplus benchmark \cite{evalplus}, Mostly Basic Python Problems (MBPP) \cite{austin2021program} for code generation and HumanEval \cite{chen2021evaluating} for code completion, and the BigCodeBench benchmark \cite{zhuo2024bigcodebench}. The client fine-tunes the encoder and the decoder on the (public) CodeAlpaca dataset \citep{codealpaca} consisting of $\sim$18k data points for Python and then evaluates \textsc{NOIR} on the (private) HumanEval dataset. We fine-tune the encoder and decoder again on the MBPP dataset as the client's private data.
In the ablation study, we enrich the CodeAlpaca dataset to $\sim$376k data points by adding code from the Stack dataset \citep{kocetkov2023the}.

\textbf{Baselines.} We consider SoTA $d_x$-privacy-preserving mechanisms for protecting user prompts, which are \textbf{T2T} \citep{10.1145/3459637.3482281,10.1145/3336191.3371856,li2023privacypreserving} and \textbf{SnD} \citep{mai2024splitanddenoise}. Defense-free mechanisms include the \textbf{fClean} model, i.e., our \textsc{NOIR} without \textsc{INDVocab} and \textsc{LTokenizer}, and the \textbf{Clean} model, i.e., the original LLM model. We use Pass@$r$ \citep{chen2021evaluating} to evaluate the model performance:
$\text{Pass@}r = \mathbb{E}_{\text{prompts}}\left[{1-\binom{n-c}{r} / \binom{n}{r}} \right],$
where $\mathbb{E}_{\text{prompts}}$ is the expectation over all prompts; for a given prompt: $r$ is the requested number of generated code versions, $n$ is the total number of generated code versions, and $c$ is the total number of correct code versions passing all unit tests. 
As in \citep{chen2021evaluating}, we set $n = 2r$. Pass@r is a significant and standard metric for quantifying model performance in generating functioning code because it evaluates the probability that at least one out of $r$ independently generated code solutions is correct \citep{chen2021evaluating}. This aligns closely with practical deployment scenarios where multiple outputs are reviewed or tested \cite{EvidentlyAI}. It is particularly valuable in code synthesis, where a single correct solution among several generated candidates is sufficient for success. Also, Pass@r encourages diversity and quality in outputs, which is crucial for robust code generation systems \cite{EvidentlyAI}. \textbf{P3T approaches} \citep{NEURIPS2023_f26119b4,wu2024privacypreserving,hong2024dpopt} \textbf{are not appropriate baselines} since they do not protect content of client's prompts and generated code as in our context.

\begin{figure}[t] 
\captionsetup[subfigure]{justification=centering}
    \centering
      \begin{subfigure}[t]{0.493\columnwidth}
        \centering
       \includegraphics[scale=0.3]{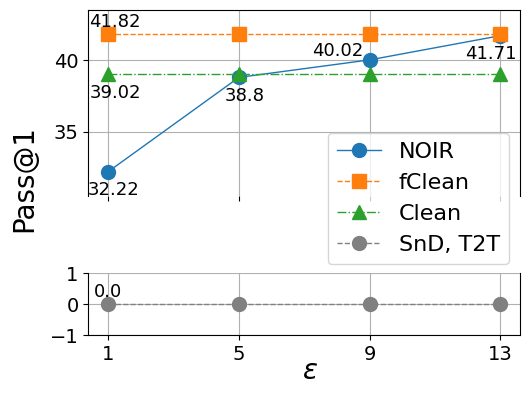} 
       \caption{MBPP dataset}
       \label{Fig:UtilityMBPP}
       \end{subfigure} 
       \hfill
      \begin{subfigure}[t]{0.493\columnwidth}
      \centering
      \includegraphics[scale=0.3]{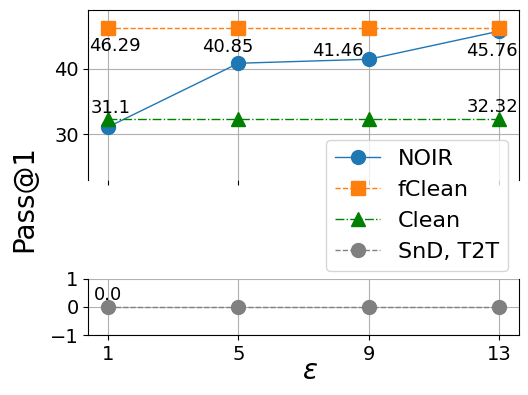} 
        \caption{HumanEval dataset}
        \label{Fig:UtilityHumanEval}
    \end{subfigure} 
    \caption{Pass@1 and IND budget $\epsilon$.} 
    \label{Fig:ModelUtilityVsEps}  
\end{figure}

\begin{figure}[t]
\captionsetup[subfigure]{justification=centering}
    \centering
      \begin{subfigure}[t]{0.493\columnwidth}
        \centering
       \includegraphics[scale=0.3]{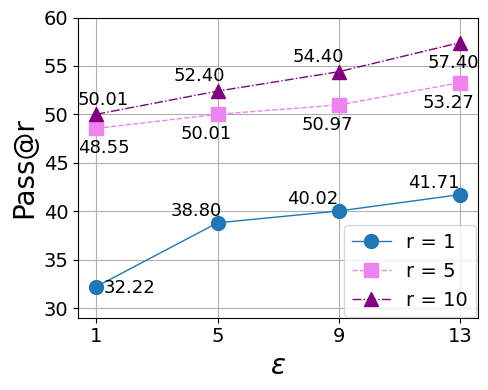} 
       \caption{MBPP dataset}
       \label{Fig:Pass@rMBPP}
       \end{subfigure} 
       \hfill
      \begin{subfigure}[t]{0.493\columnwidth}
      \centering
      \includegraphics[scale=0.3]{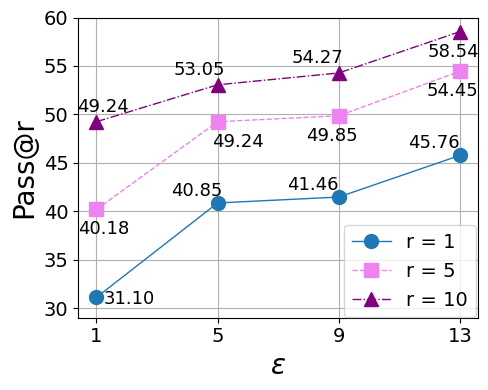}
        \caption{HumanEval dataset} 
        \label{Fig:Pass@rHuman_Eval}
    \end{subfigure}
    \caption{Pass@$r \in \{1, 5, 10\}$ and IND budget $\epsilon$.}
    \label{Fig:Pass@rVsEps}  
\end{figure}

\textbf{Q1: What is the trade-off between privacy and model performance?} 
Figure \ref{Fig:ModelUtilityVsEps} shows that SnD and T2T fail to generate code in both code generation and completion tasks across all IND budgets $\epsilon$. They only produce code-like outputs at very large $\epsilon>5,000$, but their Pass@1 remains 0, consistent with prior findings \citep{mai2024splitanddenoise}. 
In contrast, \textsc{NOIR} consistently outperforms these SoTA baselines. Higher $\epsilon$, indicating weaker IND protection, yields better Pass@1. At $\epsilon = 13$, \textsc{NOIR} achieves Pass@1 scores of 41.71 (MBPP) and 45.76 (HumanEval), with only marginal drops of $0.26\%$ and $1.14\%$ compared with upper-bounded defense-free performance. This strong performance stems from its embedding-based design, adaptive noise in \textsc{INDVocab}, and effective \textsc{STuning}. 

\textbf{Q2: What is the cost to significantly improve model performance with strong IND protection?}
With stronger protection (smaller $\epsilon$), Pass@1 performance drops more sharply (Figure \ref{Fig:ModelUtilityVsEps}). For instance, at $\epsilon = 1$ (very strong IND), 
\textsc{NOIR}'s Pass@1 scores are 32.22 (MBPP) and 31.1 (HumanEval). Fortunately, the drop can be mitigated by generating multiple code versions $r$ per prompt and selecting the best based on passed test cases. 
Even at $\epsilon = 1$, \textsc{NOIR} achieves much higher Pass@10 ($r = 10$): 50.01 (MBPP) and 49.24 (HumanEval) registering 55\% and 58.3\% improvements respectively, as shown in Figure \ref{Fig:Pass@rVsEps}. The gain grows with larger $\epsilon$, and \textsc{NOIR} Pass@10 even surpasses the defense-free Pass@1 at the cost of producing more code versions.

\textbf{Q3: What is the cost for the performance improvement?} The cost is marginal. The client generates multiple code versions of a prompt as the input. Thus, the local computational cost is $(r-1)$ multiplied by the number of test cases for each prompt to identify the best code to use, which is negligible given the lightweight decoder.
The communication and financial costs are $(r-1)$ times to generate this extra code in Pass@$r$, as a trade-off for the significant performance boost.

\textbf{Q4: How effective is \textsc{NOIR} in defending against RAs?} Figures \ref{Fig::PromptRA}, \ref{Fig::AsrCodeRA} show the privacy ASRs of PromtpRA and CodeRA under \textsc{NOIR}'s defense. On MBPP's training set without test cases in prompts, \textsc{NOIR} significantly reduces PromptRA's privacy ASR (Figure \ref{Fig::AsrMBPPTrain}) from 96\% (no test case) and 51\% (with test cases) of the defense-free mechanisms to 7\% on average at $\epsilon = 1$, achieving 92.7\% and 86.27\% improvements.
Similar reductions appear on MBPP's test set (Figure \ref{Fig::AsrMBPPTest}). For HumanEval code completion, the PromptRA's privacy ASR is 0.0, showing it cannot reconstruct a clear gist of any prompts (Figure \ref{Fig::AsrHumanEvalTest}), which is intuitive since Vec2Text \citep{morris-etal-2023-text} and BiSR \citep{chen2024unveilingvulnerabilityprivatefinetuning} were not designed for pure code prompts with \textsc{INDVocab} and \textsc{LTokenizer}.

\begin{figure}[t] 
\captionsetup[subfigure]{justification=centering}
    \centering
      \begin{subfigure}[t]{0.32\columnwidth}
        \centering
       \resizebox{1.0\textwidth}{!}{\includegraphics[scale=0.3]{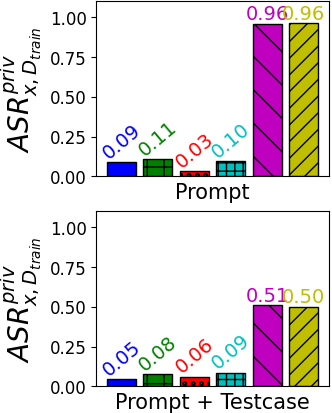}}
       \caption{$ASR^{priv}_{x, D_{train}}$ in Training Set (MBPP)}
       \label{Fig::AsrMBPPTrain}
       \end{subfigure}
       \hfill
      \begin{subfigure}[t]{0.32\columnwidth}
      \centering
      \resizebox{1.0\textwidth}{!}{\includegraphics[scale=0.3]{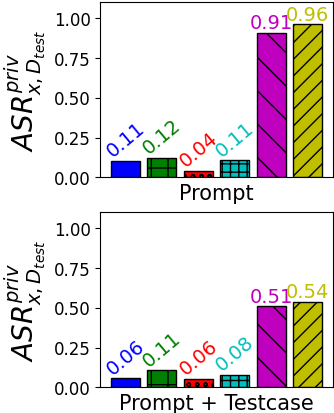}}
      \caption{$ASR^{priv}_{x, D_{test}}$ in Test Set (MBPP)}
        \label{Fig::AsrMBPPTest}
    \end{subfigure}
    \hfill
      \begin{subfigure}[t]{0.32\columnwidth}
      \centering
      \resizebox{1.0\textwidth}{!}{\includegraphics[scale=0.45]{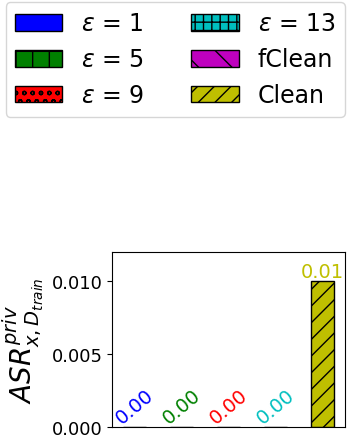}}
      \caption{$ASR^{priv}_{x, D_{test}}$ in HumanEval}
        \label{Fig::AsrHumanEvalTest}
    \end{subfigure}  
    \caption{\textsc{NOIR} against PromptRA. The lower, the better.}
    \label{Fig::PromptRA}  
\end{figure}

\begin{figure}[t] 
\captionsetup[subfigure]{justification=centering}
    \centering
    \begin{subfigure}[t]{0.999\columnwidth}
        \centering
       \includegraphics[scale=0.365]{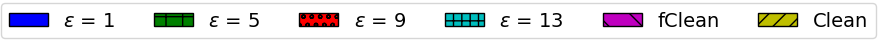}
       \end{subfigure}
    \begin{subfigure}[t]{0.32\columnwidth}
      \centering
      \resizebox{1.0\textwidth}{!}{\includegraphics[scale=1.0]{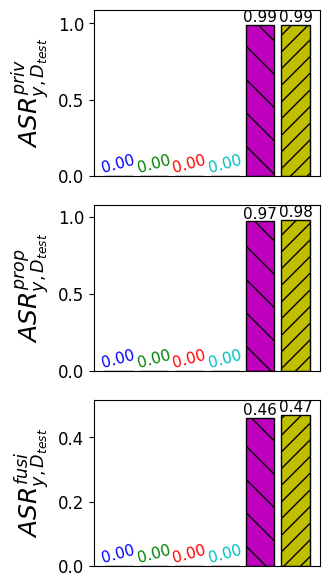}}
      \caption{MBPP Traning}
      \label{Fig::AsrMBPPTrainCode}
    \end{subfigure} \hfill
    \begin{subfigure}[t]{0.32\columnwidth}
      \centering
      \resizebox{1.0\textwidth}{!}{\includegraphics[scale=1.0]{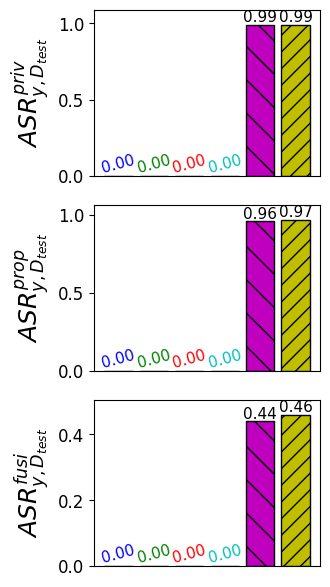}}
      \caption{MBPP Test}
      \label{Fig::AsrMBPPTestCode}
    \end{subfigure} \hfill
    \begin{subfigure}[t]{0.32\columnwidth}
      \centering
      \resizebox{1.0\textwidth}{!}{\includegraphics[scale=1.0]{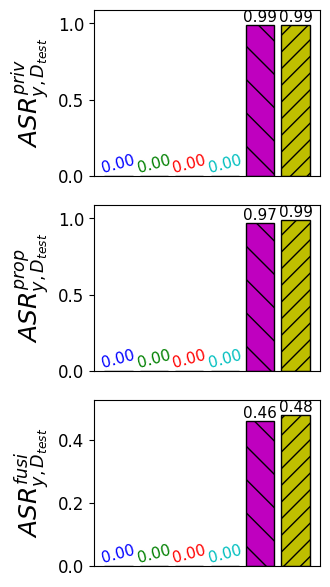}}
      \caption{HumanEval}
      \label{Fig::AsrHumanEvalCode}
    \end{subfigure}
    \caption{\textsc{NOIR} against CodeRA. The lower, the better.} 
    \label{Fig::AsrCodeRA}
\end{figure}

Although PromptRA is ineffective in code completion, CodeRA shows severe ASRs in privacy, confidentiality, and functionality (99\%, 97\%, and 46\%, respectively) against the defense-free mechanism on HumanEval (Figure \ref{Fig::AsrHumanEvalCode}). With IND protection and \textsc{LTokenizer}, \textsc{NOIR} reduces privacy and confidentiality ASRs to 0\% across IND budgets $\epsilon$, and the cloud cannot reconstruct functioning code (Fusi ASR $=$ 0.0\%). Similar results hold for code generation on MBPP training and test sets (Figures \ref{Fig::AsrMBPPTrainCode},\ref{Fig::AsrMBPPTestCode}). Under \textsc{NOIR}, $CRT_x$ and $CRT_y$ are 0.028 and 0.108 respectively, far smaller than 0.28 and 1.00 of the defense-free mechanism on MBPP. Sensitive information \texttt{leak} is 0.038 (\textsc{NOIR}), compared with 0.98 (defense-free). The reconstructed top-5 tokens prompts and code $\{$`$\langle s\rangle$', `\_', `a', `.', `to'\} and \{`per', `i', `f', `res', `p'$\}$ are meaningless. The client retains high model performance under strong IND protection with marginal cost (Q3).

\textbf{Q5: Is \textsc{NOIR} robust against frequency analysis attacks?}
In this experiment, the cloud performs sequence-based (Figure \ref{Fig:FreqAnalysis}) and token-based frequency analysis attacks, following \citep{10.1145/2810103.2813651} and \citep{Biryukov2011}. A token-based attack is a sequence-based one with the sequence length $k = 1$. Out of 32k vocabulary tokens, the cloud only recovers a few from the instruction template, common to all prompts, but none from the client prompts $x$ and code $y$. Specifically, token-based attack reconstruct 6 tokens \{`Write', `a', `Python', `function', `to', `$\langle | \text{im\_end} | \rangle$'\}, while sequence-based attacks (Figure \ref{Fig:FreqAnalysis}) recover 3 tokens \{`Write', `a', `Python'\}. After excluding these, the privacy ASR is $0.0$. Similar results hold when using public data (e.g., MBPP training set) with a similar distribution with the client's local data (e.g., MBPP test set). These findings confirm that \textsc{NOIR} resists frequency analysis attacks.

\begin{figure}[t] 
\captionsetup[subfigure]{justification=centering}
    \centering
        \begin{subfigure}[t]{0.99\columnwidth}
        \centering
        \includegraphics[scale=0.21]{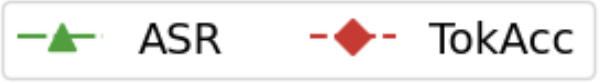}
       \end{subfigure}
        \begin{subfigure}[t]{0.31\columnwidth}
      \centering
      \includegraphics[scale=0.21]{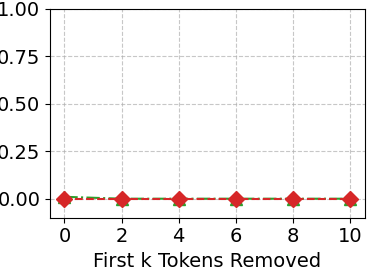}
        \caption{MBPP test using 376k dataset on sequence-based attack}
    \label{Fig:outdistopkmbppseq}
    \end{subfigure}
    \begin{subfigure}[t]{0.31\columnwidth}
      \centering
      \includegraphics[scale=0.21]{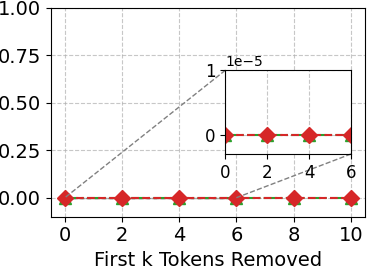}
        \caption{HumanEval using 376k dataset on sequence-based attack}
    \label{Fig:outdistopkheseq}
    \end{subfigure}
    \begin{subfigure}[t]{0.32\columnwidth}
      \centering
      \includegraphics[scale=0.19]{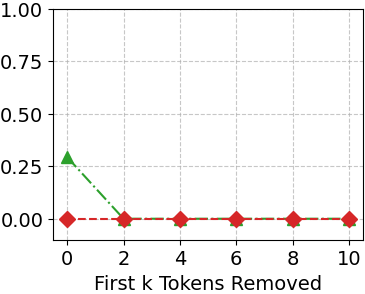}
        \caption{MBPP test using MBPP train dataset on sequence-based attack}
    \label{Fig:indistopkseq}
    \end{subfigure}
    \caption{Frequency analysis attacks ($k=3$).} 
    \label{Fig:FreqAnalysis}  
\end{figure}



\textbf{Q6: What is the model performance for different LLMs and vocabulary sizes?} We vary the LLMs by using CodeQwen1.5-7B-Chat, Llama3-8B-instruct, and CodeLlama-7B models, which differ in capabilities and vocabulary size. CodeQwen1.5-7B-Chat in \textsc{NOIR} achieves the best performance on MBPP and HumanEval (Figure \ref{Fig:VaryingTheLlms}). A gap remains between \textsc{NOIR} with CodeQwen1.5-7B-Chat and the clean model, due to insufficient training data and small privacy budgets. We address this by increasing both and analyzing the resulting RA costs.

\begin{figure}[t] 
\captionsetup[subfigure]{justification=centering}
     \centering
          \begin{subfigure}[t]{0.99\columnwidth}
            \centering
           \includegraphics[scale=0.30]{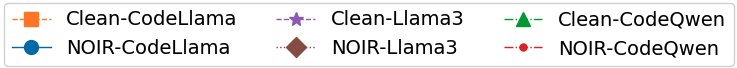}
           \end{subfigure}
    \centering
      \begin{subfigure}[t]{0.493\columnwidth}
        \centering
       \includegraphics[scale=0.3]{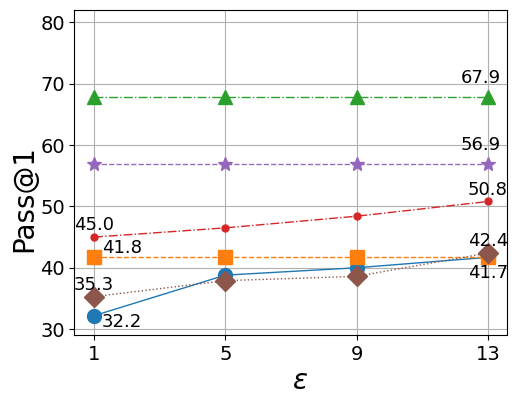} 
       \caption{MBPP dataset} 
       \label{Fig:VaryingLlmsMBPP}
       \end{subfigure}
       \hfill
      \begin{subfigure}[t]{0.493\columnwidth}
      \centering
      \includegraphics[scale=0.3]{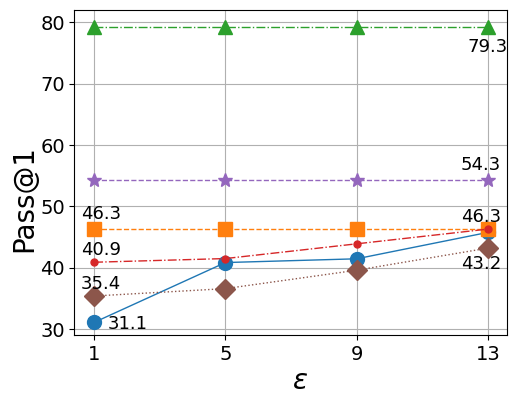} 
        \caption{HumanEval dataset}
    \label{Fig:VaryingLlmsEval}
    \end{subfigure}
    \caption{Varying LLM models.} \vspace{-2.5pt}
    \label{Fig:VaryingTheLlms}  
\end{figure}

\begin{figure}[t] 
\captionsetup[subfigure]{justification=centering}
     \centering
      \begin{subfigure}[t]{0.99\columnwidth}
        \centering
       \includegraphics[scale=0.35]{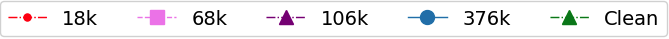}
       \end{subfigure}
    \centering
      \begin{subfigure}[t]{0.493\columnwidth}
        \centering
       \includegraphics[scale=0.3]{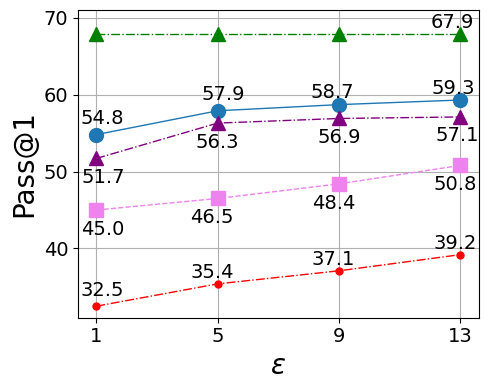}
       \caption{Enlarging the training dataset}
       \label{Fig::Qwen(EnlargeDatast)_MBPP}
       \end{subfigure}
       \hfill
       \begin{subfigure}[t]{0.493\columnwidth}
        \centering
       \includegraphics[scale=0.3]{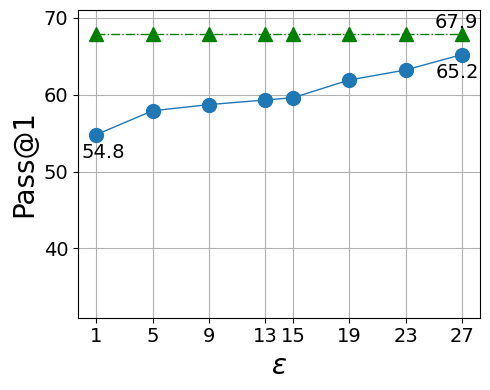}
        \caption{Increase the privacy budget $\epsilon$}
       \label{Fig:VaryingEpsMBPP}
       \end{subfigure}
       \hfill
    \caption{Enlarge the training data (CodeQwen1.5-7B-Chat).}
    \label{Fig:QwenEnlargeDataset}  
\end{figure}

\begin{figure}[t] 
\captionsetup[subfigure]{justification=centering}
    \centering
      \begin{subfigure}[t]{0.97\columnwidth}
        \centering
       \includegraphics[scale=0.32]{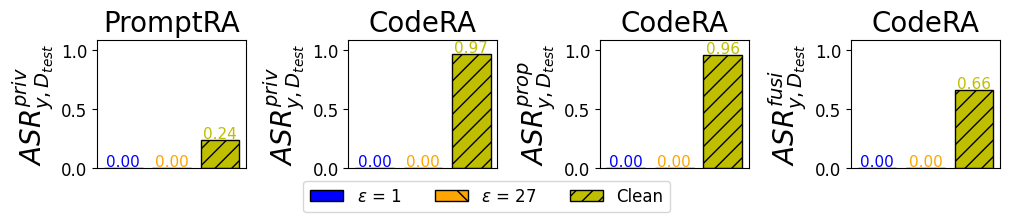} 
       \caption{\textsc{NOIR}}
       \label{Fig::ASRNOIR}
       \end{subfigure}
      \begin{subfigure}[t]{0.97\columnwidth}
        \centering
       \includegraphics[scale=0.32]{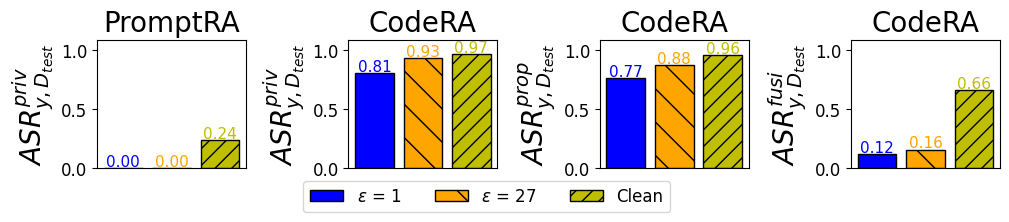}
       \caption{\textsc{NOIR} without \textsc{LTokenizer}}
       \label{Fig::ASRnoLTokenizer}
       \end{subfigure}
       \hfill
             \begin{subfigure}[t]{0.97\columnwidth}
        \centering
       \includegraphics[scale=0.32]{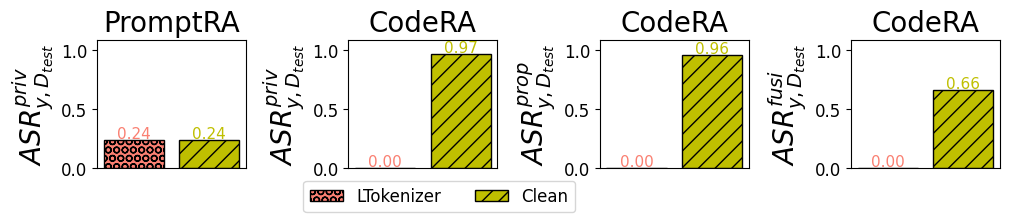}
       \caption{\textsc{NOIR} without \textsc{INDVocab}}
       \label{Fig::ASRnoLDBVocab}
       \end{subfigure} 
    \caption{Enlarging training data to 370k data points and IND $\epsilon$ to 27 under PrompRA and CodeRA.}
    \label{Fig:LDPVocab vs LTokenizer}  
\end{figure}


\textbf{Q7: What is the effect of enlarging the training data and privacy budget?} When clients (e.g., small enterprises) lack sufficient training data, they can augment it with publicly available data without disclosing the augmented set. For instance, we gradually expand CodeAlpaca from from $\sim$18k to $\sim$68k, $\sim$106k, and $\sim$376k data points by adding Python code from the Stack dataset \citep{kocetkov2023the}. As Figure \ref{Fig::Qwen(EnlargeDatast)_MBPP} shows, \textsc{NOIR}'s Pass@1 significantly improves from 39.2 to 59.3 on MBPP, reducing the gap to the Clean model (CodeQwen1.5-7B-Chat) from 42.3\% to 12.7\% at $\epsilon = 13$. With 376k data points, raising $\epsilon$ further improves Pass@1 from 54.8 to 65.2, nearly matching the defense-free fClean at 67.9 (a marginal 3.9\% drop; Figure \ref{Fig:VaryingEpsMBPP}). On HumanEval, Pass@1 rises from 46.3 to 65.9 as $\epsilon$ increases from 13 to 27. On BigCodeBench ``Instruct'', \textsc{NOIR}'s Pass@1 achieves 31.1 versus 31.7 for the clean model at $\epsilon = 27$, a marginal $1.8\%$ drop. Notably, \textsc{NOIR} remains robust against CodeRA and PromptRA with ASRs $= 0.0$.

\textbf{Q8: What is the cost of enlarging the training data and privacy budget regarding RAs?} Figure \ref{Fig::ASRNOIR} shows that increasing the privacy budget and training data does not raise reconstruction ASRs ($\sim 0.0$). This is because the task instruction $\pi$ of the CodeQwen1.5-7B-Chat model combined with the prompt $x$ is longer and involves more tokens in the vocabulary than in CodeLlama-7B, reducing the upper-bounded reconstruction probability (Section \ref{Reconstruction Upper Bounds}). Therefore, CodeQwen1.5-7B-Chat is more secure under \textsc{INDVocab} and \textsc{LTokenizer} protection. 
Hence, privacy budget increases should be paired with larger training data to balance performance and privacy.

\textbf{Q9: What are the roles of \textsc{INDVocab} and \textsc{LTokenizer} against RAs?} This experiment highlights the defensive roles of these components. CodeRA becomes very effective when the client disables \textsc{LTokenizer} and uses the open-source tokenizer (Figure \ref{Fig::ASRnoLTokenizer}). PromptRA takes its turn to be effective when the client disables \textsc{INDVocab} and uses the open-source vocabulary (Figure \ref{Fig::ASRnoLDBVocab}). Thus, \textsc{INDVocab} and \textsc{LTokenizer} defend against PromtpRA and CodeRA, respectively. Combining them is crucial for producing high-quality code while resisting attacks (Figure \ref{Fig::ASRNOIR}).

\textbf{Q10: What is the effectiveness of \textsc{STuning}, including LoRA and zeroth-order (ZO) optimization?} Training only the decoder or a shortcut model (concatenating the encoder and the decoder while removing the middle block $\theta_m$) yields poor performance (Pass@1 $=$ 0.0 at $\epsilon = 13$) with CodeLlama-7B. This is expected: decoder-only training cannot counter IND-preserving noise amplified in the large latent space of $\theta_m$, whie shortcut-only training and inference ignores key features learned by $\theta_m$. Hence, integrating encoder, decoder, and $\theta_m$ in \textsc{STuning} and inference is essential for good utility.

Applying LoRA to the encoder and decoder yields little benefit due to their lightweight design. Instead, we fine-tune the encoder using the SoTA ZO optimization \citep{zhang2024revisitingzerothorderoptimizationmemoryefficient}, which approximates gradients via enriched embeddings. However, ZO requires 50 extra inferences per data point per training iteration and offers no notable performance gain. 

\textbf{Q11: What is the client cost?}
We use an A100-80GB GPU to measure the client's cost of \textsc{NOIR} fine-tuning on CodeQwen1.5-7B-Chat across datasets from 18k to 376k data points. 
As shown in Figure \ref{Fig:ClientCommunicationCost}, GPU memory usage, fine-tuning time, and equivalent AWS hosting costs grow sub-linearly with dataset size, ensuring scalability without excessive client communication cost. A 20.89x dataset size increase raises GPU memory usage only 1.36x (49.92GB $\rightarrow$ 67.70GB; Figure \ref{Fig:ClientOverheadGB}), while AWS fees remain affordable for small enterprises (\$152 $\rightarrow$ \$2,394; Figure \ref{Fig:clientOverheadCost.png}). Therefore, our system scales efficiently with minimal cost.


\begin{figure}[t] 
\captionsetup[subfigure]{justification=centering}
    \centering
      \begin{subfigure}[t]{0.403\columnwidth}
        \centering
       \includegraphics[scale=0.375]{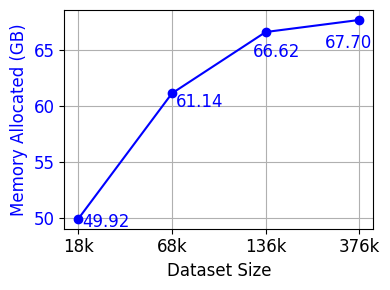} 
       \caption{GPU Memory} \vspace{-2.5pt}
       \label{Fig:ClientOverheadGB}
       \end{subfigure}
       \hfill
      \begin{subfigure}[t]{0.54\columnwidth}
      \centering
      \includegraphics[scale=0.38]{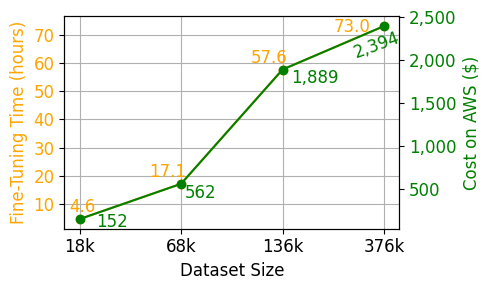}  
        \caption{Time and Cost on AWS} \vspace{-2.5pt}
    \label{Fig:clientOverheadCost.png}
    \end{subfigure} 
    \caption{Client's fine-tuning cost on 1 A100 GPU (80GB).}
    \label{Fig:ClientCommunicationCost}  
\end{figure}

\textbf{Q12: Can we reduce the cloud cost?} The cloud has the option not to fine-tune the $\theta_m$'s LoRA to reduce its computational complexity with a marginal drop $\sim$2\% of model utility as a trivial trade-off.

\textbf{Q13: What are the best numbers of attention blocks?} We set $\epsilon = 13$ and study how the number of attention blocks in the encoder and decoder affects performance.
Fixing the encoder to one block, increasing the number of decoder's attention blocks improves \textsc{NOIR}'s performance.
In contrast, adding attention blocks to the encoder does not help, as privacy noise propagates further into the latent space, limiting the cloud ability to enrich prompt embeddings and making encoder/decoder fine-tuning difficult.

\textbf{Q14: What is the NOIR's end-to-end latency?}
To evaluate NOIR's latency, we deploy the system with Qwen2.5-Coder-32B-Instruct when the encoder/decoder and middle block servers are in the same cloud region. Each server is equipped with a single A100 80GB GPU. For comparison, we also deploy the fully connected model on a single server with the same GPU achieving 22.96 tokens/s, whereas NOIR's latency is 21.49 tokens/s. These results indicate that when the client and cloud servers are co-located in a close proximity, NOIR incurs only a marginal latency overhead.

\section{Discussion}

In this section, we discuss limitations in terms of security auditing and practicality of \textsc{NOIR} and potential solutions.

\textbf{Practicality.} LLM vendors may not have strong incentives/motivation to adopt split learning. However, third-party providers (independent of LLM creators) can leverage existing open-source LLMs to provide \textsc{NOIR}, which will allow them to charge data security and IP sensitive enterprises for this privacy-preserving service. \textsc{NOIR} requires seamless integration between client-side (lightweight encoder/decoder) and cloud-hosted (middle block) components to make this model practical.

\textbf{Security Auditing.} While stronger security auditing mechanisms remain unavailable, adversaries could exploit advanced adaptive attacks to undermine $\epsilon$-IND. A powerful adversary might control the cloud-side LoRA to fine-tune middle blocks, reducing or eliminating noise effects. By analyzing the distribution of IND-preserved embeddings across multiple prompts, they could model noise patterns and dynamically adjust LoRA to improve reconstruction accuracy. A potential defense is to explore privacy-preserving LoRA protocols, training the cloud’s LoRA via zero-knowledge proofs or trusted execution environments to safeguard client data patterns. Such solutions may not be cost effective, and therefore further research is necessary to explore this topic.

Cross-embedding similarity could reveal semantic relationships between tokens, while token clustering across prompts might infer the original token’s semantic role, such as identifying “function” from consistent contextual embeddings like “def” or “defining a function.” Also, cross-prompts clustering could uncover semantic structures, enabling attackers to bypass $\epsilon$-IND through pattern inference. The key idea of these adaptive attacks is to discover frequent sequences of token embeddings then match them with frequent sequences of tokens extracted from public data. To defend against these sophisticated adversaries, \textsc{NOIR} will prioritize enhancing its privacy mechanisms. This includes developing adaptive strategies, such as context-dependent \textbf{dynamic prefix tokens} using encoder attention block position embeddings to alter tokenization during training or inference, thereby preventing the cloud from learning semantic patterns.

In a pilot study, we considered these cross-embedding-based adaptive attacks by clustering token embeddings based on 1) their similarity, 2) their positions, and 3) both their similarity and positions using public data (MBPP training set) on the Qwen2.5-Coder-32B-Instruct and Qwen3-32B models to identify common sequences across prompts. The position-based clustering attack (i.e., sequences of token embeddings are matched with sequence of tokens in the public data based on their positions in the prompts) shows the most promising result by recovering $\epsilon$-IND-preserving embeddings ($\epsilon=27$) of a few persistent tokens \{`Python', `function', `a', `write', `to'\} at fixed sequence patterns across prompts in the system prompt template. Only a few tokens are identified because the encoder attention block position embeddings disrupt the token embeddings, making it difficult to exploit the token embedding similarity. Randomly adding a dynamic prefix up to 8 random characters, e.g., <A@1p?He5>, to the system and instruction prompts prevents this marginal leak of the $\epsilon$-IND-preserving token embeddings by further disrupting the embedding sequence patterns (in terms of both similarity and positions) across prompts without affecting model performance given the model's strong reasoning capability (Appx. \ref{General Text}). Also, overly strict protection, such as increasing the dynamic prefix token to 20 characters, can notably degenerate the model performance.
Further research is needed to bolster \textsc{NOIR}’s defense against adaptive and cross-prompts attacks, ensuring robust protection in practice.

\textbf{Text Generation.} \textsc{NOIR} is limited to private code generation due to its unique challenge, fine-tuning data, and pipeline rather than general text generation. Extending \textsc{NOIR} to general text generation involves adapting the vocabulary and tokenizer to handle diverse natural language tasks, modifying the training data and architecture to support text generation, and ensuring robust privacy protections against reconstruction attacks. Our pilot study shows that \textsc{NOIR} achieves highly competitive reasoning performance using Qwen open-source LLMs (Appx. \ref{General Text}). A \textsc{NOIR}'s trial system for general text generation is available: \url{https://noir.oppyai.com}.

Several other research directions include:
\textbf{(1) Embedding Alignment:} The cloud must avoid responding to malicious embeddings from harmful prompts (e.g., ``generating malware code''), a harder task than detecting harmful prompts.
\textbf{(2) Proprietary LLMs:} Model architecture, vocabulary, and token embedding size are sensitive to the proprietary LLM-provider. A solution is for clients to use open-source encoders/decoders, while the cloud aligns embedding sizes with its proprietary LLM via feature transformation (e.g., NVIB \cite{henderson2023a}). Pilot results \hai{(Appx. \ref{NVIBPrelim})} show promising performance on benchmarks.

\section{Conclusion}
\label{Conclusion and Discussion}

This paper presents \textsc{NOIR}, a code generation framework that protects clients' prompts and code using open-source LLMs. Clients send prompt embeddings derived from a local encoder to the cloud-hosted model and receive enriched embeddings, which are decoded locally to generate code. To prevent cloud reconstruction of prompts and code, we introduce a novel concept of $\epsilon$-IND-preserving vocabulary (\textsc{INDVocab}) with a randomized tokenizer (\textsc{LTokenizer}), and propose \textsc{STuning}, a client-side fine-tuning method for the encoder and decoder. 
Extensive theoretical and empirical results show that \textsc{NOIR} significantly outperforms SoTA baselines while safeguarding data privacy, code confidentiality, and code functionality against reconstruction and frequency analysis attacks from an honest-but-curious cloud. To demonstrate practicality, we open-source \textsc{NOIR} as a privacy-preserving coding agent using a Qwen2.5-Coder-32B-Instruct model with Pass@1: 83.6 (MBPP), 85.4 (HumanEval), and 47.2 (BigCodeBench).

\section*{Acknowledgment}
This work was supported by the National Science Foundation (NSF) grants CNS-(1935928, 1935923), CNS 2237328, and DGE 2043104. We thank Thai Nguyen for the study in Appx. \ref{NVIBPrelim} and the dedicated reviewers who help improve \textsc{NOIR}.

\appendix
\section*{Ethical Considerations}

Our research offers the potential to benefit all stakeholders, including: 1) Individuals and enterprises, who use LLMs as software coding tools; and 2) LLM-providers. No stakeholders will be adversely impacted by the publication of your research now and in the future, since we develop a solution to protect IP and data security in prompting LLMs for private code generation without affecting any currently available commercial systems.

\section*{Open Science}

We open-source \textsc{NOIR} ({\color{blue}\url{https://tinyurl.com/NOIR-Artifact}}) based on Qwen2.5-Coder-32B-Instruct and release its API through a privacy-preserving coding agent for the public use in a web-service ({\color{blue}\url{https://noir.oppyai.com}}) and a Visual Studio extension ({\color{blue}\url{https://tinyurl.com/NOIR-Artifact}}), integrated directly into the software development pipeline. The source-code is available with detailed experimental configurations for reproducibility purposes.

\bibliography{references}

\begin{thebibliography}{69}
\providecommand{\natexlab}[1]{#1}
\providecommand{\url}[1]{\texttt{#1}}
\expandafter\ifx\csname urlstyle\endcsname\relax
  \providecommand{\doi}[1]{doi: #1}\else
  \providecommand{\doi}{doi: \begingroup \urlstyle{rm}\Url}\fi

\bibitem[Github(2023)]{copilot}
Github.
\newblock Research: Quantifying github copilot’s impact on code quality, 2023.

\bibitem[CNBC(2024)]{BiggestRiskGenAI}
CNBC.
\newblock The biggest risk corporations see in gen ai usage isn’t hallucinations, 2024.

\bibitem[T-Minus365(2024)]{SecurityRisksCopilot}
T-Minus365.
\newblock Microsoft 365 copilot | security risks \& how to protect your data, 2024.
\newblock [accessed 17-June-2024].

\bibitem[Concentric.ai(2024)]{TooMuchAccess}
Concentric.ai.
\newblock Too much access? microsoft copilot security concerns explained, 2024.
\newblock [accessed 17-June-2024].

\bibitem[Journal(2025)]{SEJ}
Search~Engine Journal.
\newblock Openai is pulling shared chatgpt chats from google search, 2025.

\bibitem[Mag(2023)]{SamsungIncident}
PC~Mag.
\newblock Samsung software engineers busted for pasting proprietary code into chatgpt, 2023.

\bibitem[Insights(2024)]{Market2032}
Fortune~Business Insights.
\newblock Low code development platform market, 2024.
\newblock [Online; accessed 17-June-2024].

\bibitem[Xiao et~al.(2023)Xiao, Lin, and Han]{xiao2023offsitetuning}
Guangxuan Xiao, Ji~Lin, and Song Han.
\newblock Offsite-tuning: Transfer learning without full model, 2023.

\bibitem[Medium(2025)]{MediumOpenSourceLLM}
Medium.
\newblock Hybrid llm strategies for enterprises: Balancing self-hosted governance with cloud flexibility, 2025.

\bibitem[Wang(2025)]{OpenSourceLLM}
Chenxi Wang.
\newblock Do enterprises want open-weight models?, 2025.

\bibitem[Network(2025)]{paloaltonet}
Paloalto Network.
\newblock Why self-managed ai models are blind spots and what to do about it, 2025.

\bibitem[Trade(2025)]{DigitalTrade}
Digital Trade.
\newblock Using self-hosted large language models (llms) securely in government, 2025.

\bibitem[Forbes(2023)]{GenAIDataCenter}
Forbes.
\newblock Generative ai breaks the data center: Data center infrastructure and operating costs projected to increase to over 76 billion by 2028, 2023.

\bibitem[Physics(2024)]{AIDataCenter}
Construction Physics.
\newblock How to build an ai data center, 2024.
\newblock [Online; accessed 17-June-2024].

\bibitem[Duan and et~al.(2023)]{NEURIPS2023_f26119b4}
Haonan Duan and et~al.
\newblock Flocks of stochastic parrots: Differentially private prompt learning for large language models.
\newblock In \emph{NeurIPS}, 2023.

\bibitem[Wu and et~al.(2024)]{wu2024privacypreserving}
Tong Wu and et~al.
\newblock Privacy-preserving in-context learning for large language models.
\newblock In \emph{ICLR}, 2024.

\bibitem[Hong and et~al.(2024)]{hong2024dpopt}
Junyuan Hong and et~al.
\newblock {DP}-{OPT}: Make large language model your privacy-preserving prompt engineer.
\newblock In \emph{ICLR}, 2024.

\bibitem[Tang and et~al.(2024)]{tang2024privacypreserving}
Xinyu Tang and et~al.
\newblock Privacy-preserving in-context learning with differentially private few-shot generation.
\newblock In \emph{ICLR}, 2024.

\bibitem[Qu and et~al.(2021)]{10.1145/3459637.3482281}
Chen Qu and et~al.
\newblock Natural language understanding with privacy-preserving bert.
\newblock In \emph{CIKM}, 2021.

\bibitem[Feyisetan and et~al.(2020)]{10.1145/3336191.3371856}
Oluwaseyi Feyisetan and et~al.
\newblock Privacy- and utility-preserving textual analysis via calibrated multivariate perturbations.
\newblock In \emph{WSDM}, 2020.

\bibitem[Li et~al.(2023{\natexlab{a}})Li, Tan, and Liu]{li2023privacypreserving}
Yansong Li, Zhixing Tan, and Yang Liu.
\newblock Privacy-preserving prompt tuning for large language model services, 2023{\natexlab{a}}.

\bibitem[Mai and et~al.(2024)]{mai2024splitanddenoise}
Peihua Mai and et~al.
\newblock Split-and-denoise: Protect large language model inference with local differential privacy.
\newblock In \emph{ICML}, 2024.

\bibitem[Chatzikokolakis and et~al.()]{10.1007/978-3-642-39077-7_5}
Konstantinos Chatzikokolakis and et~al.
\newblock Broadening the scope of differential privacy using metrics.
\newblock In \emph{PETS'13}.

\bibitem[Morris and et~al.(2023)]{morris-etal-2023-text}
John Morris and et~al.
\newblock Text embeddings reveal (almost) as much as text.
\newblock In \emph{EMNLP}, 2023.

\bibitem[Chen and et~al.(2024)]{chen2024unveilingvulnerabilityprivatefinetuning}
Guanzhong Chen and et~al.
\newblock Unveiling the vulnerability of private fine-tuning in split-based frameworks for large language models: A bidirectionally enhanced attack.
\newblock In \emph{CCS}, 2024.

\bibitem[Cycode(2025)]{sourcecodeleaks}
Cycode.
\newblock Top source code leaks 2020-2025, 2025.

\bibitem[Chen and et~al.(2021)]{chen2021evaluating}
Mark Chen and et~al.
\newblock Evaluating large language models trained on code.
\newblock \emph{arXiv:2107.03374}, 2021.

\bibitem[Zhuo and et~al.(2024)]{zhuo2024bigcodebench}
Terry~Yue Zhuo and et~al.
\newblock Bigcodebench: Benchmarking code generation with diverse function calls and complex instructions.
\newblock \emph{arXiv preprint arXiv:2406.15877}, 2024.

\bibitem[Skean and et~al.()]{skean2025layer}
Oscar Skean and et~al.
\newblock Layer by layer: Uncovering hidden representations in language models.
\newblock In \emph{ICML'25}.

\bibitem[Vepakomma and et~al.(2018)]{DBLP:journals/corr/abs-1812-00564}
Praneeth Vepakomma and et~al.
\newblock Split learning for health: Distributed deep learning without sharing raw patient data.
\newblock \emph{CoRR}, 2018.

\bibitem[Lin and et~al.(2024)]{10529950}
Zheng Lin and et~al.
\newblock Split learning in 6g edge networks.
\newblock \emph{IEEE Wireless Communications}, 2024.

\bibitem[Pasquini et~al.(2021)Pasquini, Ateniese, and Bernaschi]{10.1145/3460120.3485259}
Dario Pasquini, Giuseppe Ateniese, and Massimo Bernaschi.
\newblock Unleashing the tiger: Inference attacks on split learning.
\newblock CCS, 2021.

\bibitem[Hu and et~al.(2022)]{hu2022lora}
Edward~J Hu and et~al.
\newblock Lo{RA}: Low-rank adaptation of large language models.
\newblock In \emph{ICLR}, 2022.

\bibitem[Austin and et~al.(2021)]{austin2021program}
Jacob Austin and et~al.
\newblock Program synthesis with large language models.
\newblock \emph{arXiv:2108.07732}, 2021.

\bibitem[Erlingsson et~al.(2014)Erlingsson, Pihur, and Korolova]{10.1145/2660267.2660348}
\'{U}lfar Erlingsson, Vasyl Pihur, and Aleksandra Korolova.
\newblock Rappor: Randomized aggregatable privacy-preserving ordinal response.
\newblock In \emph{CCS}, 2014.

\bibitem[AI(2024)]{MistralAITokenizer}
Mistral AI.
\newblock Tokenization, 2024.

\bibitem[Dwork et~al.(2014)]{dwork2014}
Cynthia Dwork et~al.
\newblock The algorithmic foundations of differential privacy.
\newblock \emph{Found. and Trends in Theoretical CS}, 2014.

\bibitem[Kasiviswanathan and et~al.(2008)]{4690986}
Shiva~Prasad Kasiviswanathan and et~al.
\newblock What can we learn privately?
\newblock In \emph{FOCS}, pages 531--540, 2008.

\bibitem[Duchi et~al.(2013{\natexlab{a}})Duchi, Jordan, and Wainwright]{6686179}
John~C. Duchi, Michael~I. Jordan, and Martin~J. Wainwright.
\newblock Local privacy and statistical minimax rates.
\newblock In \emph{FOCS}, pages 429--438, 2013{\natexlab{a}}.

\bibitem[Ni and et~al.(2021)]{ni2021largedualencodersgeneralizable}
Jianmo Ni and et~al.
\newblock Large dual encoders are generalizable retrievers, 2021.

\bibitem[Li and et~al.(2023)]{li2023towards}
Zehan Li and et~al.
\newblock Towards general text embeddings with multi-stage contrastive learning.
\newblock \emph{arXiv:2308.03281}, 2023.

\bibitem[Song and Raghunathan(2020)]{10.1145/3372297.3417270}
Congzheng Song and Ananth Raghunathan.
\newblock Information leakage in embedding models.
\newblock In \emph{CCS}, 2020.

\bibitem[Li et~al.(2023{\natexlab{b}})Li, Xu, and Song]{li-etal-2023-sentence}
Haoran Li, Mingshi Xu, and Yangqiu Song.
\newblock Sentence embedding leaks more information than you expect: Generative embedding inversion attack to recover the whole sentence.
\newblock In \emph{ACL}, July 2023{\natexlab{b}}.

\bibitem[Biryukov(2011)]{Biryukov2011}
Alex Biryukov.
\newblock \emph{Codebook Attack}.
\newblock Springer US, 2011.

\bibitem[Papineni and et~al(2002)]{papineni-etal-2002-bleu}
Kishore Papineni and et~al.
\newblock {B}leu: a method for automatic evaluation of machine translation.
\newblock ACL, 2002.

\bibitem[Lin(2004)]{lin-2004-rouge}
Chin-Yew Lin.
\newblock {ROUGE}: A package for automatic evaluation of summaries.
\newblock In \emph{Text Sum. Branches Out}, 2004.

\bibitem[Cloud()]{BleuScoreRange}
Google Cloud.
\newblock Evaluating models.
\newblock [Online; accessed 17-June-2024].

\bibitem[klu.ai()]{ROUGEScoreRange}
klu.ai.
\newblock What is the rouge score (recall-oriented understudy for gisting evaluation)?

\bibitem[Ren and et~al.(2020)]{CodeBleu}
Shuo Ren and et~al.
\newblock Codebleu: a method for automatic evaluation of code synthesis.
\newblock \emph{CoRR}, 2020.

\bibitem[Chaudhary(2023)]{codealpaca}
Sahil Chaudhary.
\newblock Code alpaca: An instruction-following llama model for code generation, 2023.

\bibitem[Raffel and et~al.(2023)]{raffel2023exploringlimitstransferlearning}
Colin Raffel and et~al.
\newblock Exploring the limits of transfer learning with a unified text-to-text transformer.
\newblock \emph{arXiv:1910.10683}, 2023.

\bibitem[Mattern et~al.(2022)Mattern, Weggenmann, and Kerschbaum]{mattern-etal-2022-limits}
J.~Mattern, B.~Weggenmann, and F.~Kerschbaum.
\newblock The limits of word level differential privacy.
\newblock ACL, 2022.

\bibitem[Warner(1965)]{warner1965randomized}
S.~L. Warner.
\newblock Randomized response: A survey technique for eliminating evasive answer bias.
\newblock \emph{JASA}, 1965.

\bibitem[Duchi et~al.(2013{\natexlab{b}})Duchi, Jordan, and Wainwright]{duchi2013local}
J.~C. Duchi, M.~I. Jordan, and M.~J. Wainwright.
\newblock Local privacy and statistical minimax rates.
\newblock In \emph{IEEE Protocols for secure computations}, pages 429--438, 2013{\natexlab{b}}.

\bibitem[Duchi and Rogers(2019)]{duchi2019lower}
J.~Duchi and R.n Rogers.
\newblock Lower bounds for locally private estimation via communication complexity.
\newblock In \emph{COLT}, pages 1161--1191, 2019.

\bibitem[Wang and et~al.()]{wang2019collecting}
N.~Wang and et~al.
\newblock Collecting and analyzing multidimensional data with local differential privacy.
\newblock In \emph{ICDE'19}.

\bibitem[Zhao and et~al.(2020)]{zhao2020local}
Y.~Zhao and et~al.
\newblock Local differential privacy based federated learning for internet of things.
\newblock \emph{IEEE IoT-J}, 2020.

\bibitem[Lyu et~al.(2020)Lyu, Li, He, and Xiao]{lyu2020towards}
L.~Lyu, Y.~Li, X.~He, and T.~Xiao.
\newblock Towards differentially private text representations.
\newblock In \emph{ACM SIGIR}, 2020.

\bibitem[Nguyen et~al.(2023)Nguyen, Lai, Phan, and Thai]{Nguyen_Lai_Phan_Thai_2023}
T.~Nguyen, P.~Lai, N.~Phan, and M.~T. Thai.
\newblock Xrand: Differentially private defense against explanation-guided attacks.
\newblock \emph{AAAI}, 2023.

\bibitem[Liu and et~al.(2023)]{evalplus}
J.~Liu and et~al.
\newblock Is your code generated by chat{GPT} really correct? rigorous evaluation of large language models for code generation.
\newblock In \emph{NeurIPS}, 2023.

\bibitem[Kocetkov and et~al.(2023)]{kocetkov2023the}
Denis Kocetkov and et~al.
\newblock The stack: 3 {TB} of permissively licensed source code.
\newblock \emph{TMLR}, 2023.

\bibitem[EvidentlyAI(2025)]{EvidentlyAI}
EvidentlyAI.
\newblock 30 llm evaluation benchmarks and how they work, 2025.

\bibitem[Naveed and et~al.(2015)]{10.1145/2810103.2813651}
M.~Naveed and et~al.
\newblock Inference attacks on property-preserving encrypted databases.
\newblock CCS, 2015.

\bibitem[Zhang and et~al.(2024)]{zhang2024revisitingzerothorderoptimizationmemoryefficient}
Yihua Zhang and et~al.
\newblock Revisiting zeroth-order optimization for memory-efficient llm fine-tuning: A benchmark.
\newblock In \emph{ICML}, 2024.

\bibitem[Henderson and Fehr(2023)]{henderson2023a}
James Henderson and Fabio~James Fehr.
\newblock A {VAE} for transformers with nonparametric variational information bottleneck.
\newblock In \emph{ICLR}, 2023.

\bibitem[Dingledine et~al.(2004)Dingledine, Mathewson, and Syverson]{269582}
R.~Dingledine, N.~Mathewson, and P.~Syverson.
\newblock Tor: The {Second-Generation} onion router.
\newblock USENIX, 2004.

\bibitem[Vaswani and et~al.(2017)]{NIPS2017_3f5ee243}
Ashish Vaswani and et~al.
\newblock Attention is all you need.
\newblock NeurIPS, 2017.

\bibitem[Zheng and et~al.(2025)]{zheng2025group}
Chujie Zheng and et~al.
\newblock Group sequence policy optimization.
\newblock \emph{arXiv preprint arXiv:2507.18071}, 2025.

\bibitem[White and et~al.(2024)]{white2024livebench}
Colin White and et~al.
\newblock Livebench: A challenging, contamination-limited llm benchmark.
\newblock \emph{arXiv preprint arXiv:2406.19314}, 2024.

\end{thebibliography}


\appendix

\vspace{-10pt}
\section{Proof of Theorem \ref{theorem-beta-bound}}

\begin{proof}
Given two possible values $e^i_t$ and $e^{i}_{t}$$'$ of the $i^\text{th}$-feature in the token embedding $e_t$, i.e., $e^i_t, e^{i}_{t}$$' \in \{e^i_l\}_{l \in V}$, and any possible output $z \in Range(ARR)$, where $Range(ARR)$ denotes every possible output of ARR, we have:
\begin{footnotesize}
\begin{align}
\frac{P(ARR(e^i_t) = z )}{P(ARR({e^{i}_{t}}') = z)} \le \frac{\max P(ARR(e^i_t) = z)}{\min P(ARR({e^{i}_{t}}') = z)} = \frac{\max P(ARR(e^i_t) = z)}{\min P(ARR(e^i_k) = z)}
 \label{proof-beta}
\end{align}
\end{footnotesize}

Following the typical definition of RR mechanisms \citep{warner1965randomized}, the probability $p_i$ must be larger than or equal to any probabilities $q_{i, k}$. We assume having this condition: 
\begin{equation}
\forall k \in V \setminus t: p_i \geq q_{i, k}. 
\label{PandQCondition}
\end{equation}

From Eqs. \ref{proof-beta} and \ref{PandQCondition}, we have that
\begin{small}
\begin{align} 
  &\frac{P(ARR(e^{i_t}) = z )}{P(ARR({e^{i}_{t}}') = z)} \leq \frac{ \frac{\exp({\beta_i})}{\exp({\beta_i}) + |V| -1} }{  \min (\frac{\exp(-\Delta^{i_{t,k}}/m)}{ \sum_{l \in V} \exp(-\Delta^{i_{t,l}}/m)} \frac{|V| - 1}{\exp({\beta_i}) + |V| -1})  } \nonumber \\
  &= \frac{ \exp(\beta_i) }{(|V| -1) \min ( \frac{\exp(-\Delta^{i_{t,k}/m})}{ \sum_{l \in V} \exp(-\Delta^i_{t,l}/m)} ) } \le \exp(\varepsilon_i).
\end{align}
\end{small}

Taking a natural logarithm of Eq. \ref{proof-beta}, we obtain: $\ln ( \frac{ \exp(\beta_i) }{(|V| -1) \min ( \frac{\exp(-\Delta^i_{t,k}/m)}{ \sum_{l \in V} \exp(-\Delta^i_{t,l}/m)} ) }  ) \le \ln (\exp(\varepsilon_i))$
\begin{small}
\begin{align}
\Leftrightarrow \beta_i \le \varepsilon_i + \ln(|V|-1) + \ln (  \frac{\min\{\exp(-\Delta^i_{t,k}/m)\}_{k \in V \setminus t}}{ \sum_{l \in V} \exp(-\Delta^i_{t,l}/m)}).
  \label{Condition1}
\end{align}
\end{small}

\noindent Let us recall that, for the Eq. \ref{Condition1} to hold, we need the condition in Eq. \ref{PandQCondition} to hold. In fact, we can rewrite Eq. \ref{PandQCondition} as follows:
$p_i \geq \max \{q_{i, k}\}_{k \in V \setminus t} = \frac{|V| -1}{\exp({\beta_i})  + |V| -1} \frac{\max\{\exp(-\Delta^i_{t,k}/m)\}_{k \in V \setminus t}}{ \sum_{l \in V} \exp(-\Delta^i_{t,l}/m)}.$

\noindent This is equivalent to: $\frac{\exp({\beta_i})}{\exp({\beta_i}) + |V| -1} \geq \frac{|V| -1}{\exp({\beta_i})  + |V| -1} \frac{\max\{\exp(-\Delta^i_{t,k}/m)\}_{k \in V \setminus t}}{ \sum_{l \in V} \exp(-\Delta^i_{t,l}/m)} \Leftrightarrow \exp({\beta_i}) \geq (|V| - 1) \frac{\max\{\exp(-\Delta^i_{t,k}/m)\}_{k \in V \setminus t}}{ \sum_{l \in V} \exp(-\Delta^i_{t,l}/m)}$ 
\begin{small}
\begin{align}
\Leftrightarrow \beta_i \geq \ln (|V| - 1) + \ln (\frac{\max\{\exp(-\Delta^i_{t,k}/m)\}_{k \in V \setminus t}}{ \sum_{l \in V} \exp(-\Delta^i_{t,l}/m)}). 
\label{Condition2}
\end{align}
\end{small}

For Eqs. \ref{Condition1} and \ref{Condition2}, which represent the upper-bound and the lower-bound of $\beta_i$ respectively, to hold simultaneously, resulting in feasible $\beta_i$, we need the upper-bound to be larger than or equal to the lower-bound, as follows:
\begin{footnotesize}
\begin{align}
& \varepsilon_i + \ln(|V|-1) + \ln (  \frac{\min\{\exp(-\Delta^i_{t,k}/m)\}_{k \in V \setminus t}}{ \sum_{l \in V} \exp(-\Delta^i_{t,l}/m)}) \geq \ln(|V| - 1) + \nonumber \\
& \ln (\frac{\max\{\exp(-\Delta^i_{t,k}/m)\}_{k \in V \setminus t}}{ \sum_{l \in V} \exp(-\Delta^i_{t,l}/m)}) \Leftrightarrow \varepsilon_i \geq \ln (\frac{\max \{\exp(\Delta^i_{t,k}/m)\}_{k \in V \setminus t}}{\min \{\exp(\Delta^i_{t,k}/m)\}_{k \in V \setminus t}}) 
\label{EpsilonCond1}
\end{align}
\end{footnotesize}

\noindent Let us denote $\Delta^i_{t, min} = \min \{\Delta^i_{t, k}\}_{k \in V \setminus t}$ and $\Delta^i_{t, max} = \max \{\Delta^i_{t, k}\}_{k \in V \setminus t}$. From Eq. \ref{EpsilonCond1}, we have that
\begin{small}
\begin{align}
\varepsilon_i \geq \ln (\frac{\exp(\Delta^i_{t, max} / m)}{\exp(\Delta^i_{t, min} / m)})
\Leftrightarrow \varepsilon_i \geq \frac{1}{m}(\Delta^i_{t, max} - \Delta^i_{t, min}).
\label{lowerboundEpsilon}
\end{align}
\end{small}

\noindent Consequently, from Eqs. \ref{Condition1}, \ref{Condition2}, and \ref{lowerboundEpsilon}, Theorem \ref{theorem-beta-bound} holds. 
\end{proof}


\section{Theorem \ref{Theorem-composition}}
\label{Appendix Theorem-composition}

\begin{theorem} Applying ARR to independently randomize every $i^{\text{th}}$-feature with a privacy budget $\epsilon_i$ in a token embedding $e_t$ preserves $\epsilon$-IND, where $\epsilon = \sum_{i\in e_t}\epsilon_i$.
\label{Theorem-composition}
\end{theorem}

\begin{proof}
Given any two possible embeddings $e_t = \{e^i_t\}_{i \in e_t}$ and $e'_t = \{e^i_t$$'\}_{i \in e'_t}$ of a token $t$, i.e., $\forall i^{\text{th}}\text{-feature in } e_t: e^i_t, e^{i}_{t}$$' \in \{e^i_l\}_{l \in V}$, and any possible output $\mathcal{O} = \{z^i \in Range(ARR)\}_{\forall i^{\text{th}}\text{-feature in } e_t}$, where $Range(ARR)$ denotes every possible output of ARR,
we have:
$\frac{P(ARR(e_t) = \mathcal{O})}{P(ARR(e'_t) = \mathcal{O})} =  \le \prod_{i \in e_t} \frac{\max P(ARR(e^i_t) = z^i)}{\min P(ARR({e^i_t}') = z^i)} = \prod_{i \in e_t} \frac{\max P(ARR(e^i_t) = z^i)}{\min P(ARR(e^i_k) = z^i)} \le \prod_{i \in e_t} \exp(\epsilon_i) = \exp(\sum_{i \in e_t} \epsilon_i).$
Consequently, Theorem \ref{Theorem-composition} holds.
\end{proof}

\section{Prompt-level Protection} 
\label{Mitigating Limitations of Token-level Privacy}

Token-level $d_x$-privacy has two key weaknesses: (1) privacy budget and authorship leakage scale with the number of tokens in a sentence, and (2) protection varies across sentences with different token counts \citep{mattern-etal-2022-limits}. Authorship leakage \citep{10.1007/978-3-642-39077-7_5} is not a concern in our study, as the cloud knows the client's identity. While techniques like Tor \citep{269582} can address authorship privacy, we aim to establish an upper bound on the probability of prompt $x$ being reconstructed, thereby exploring the relationship between \textsc{INDVocab} and reconstruction risk.

\begin{wrapfigure}{r}{0.5\columnwidth}
  \begin{center} \vspace{-10pt}
\includegraphics[width=0.5\columnwidth]{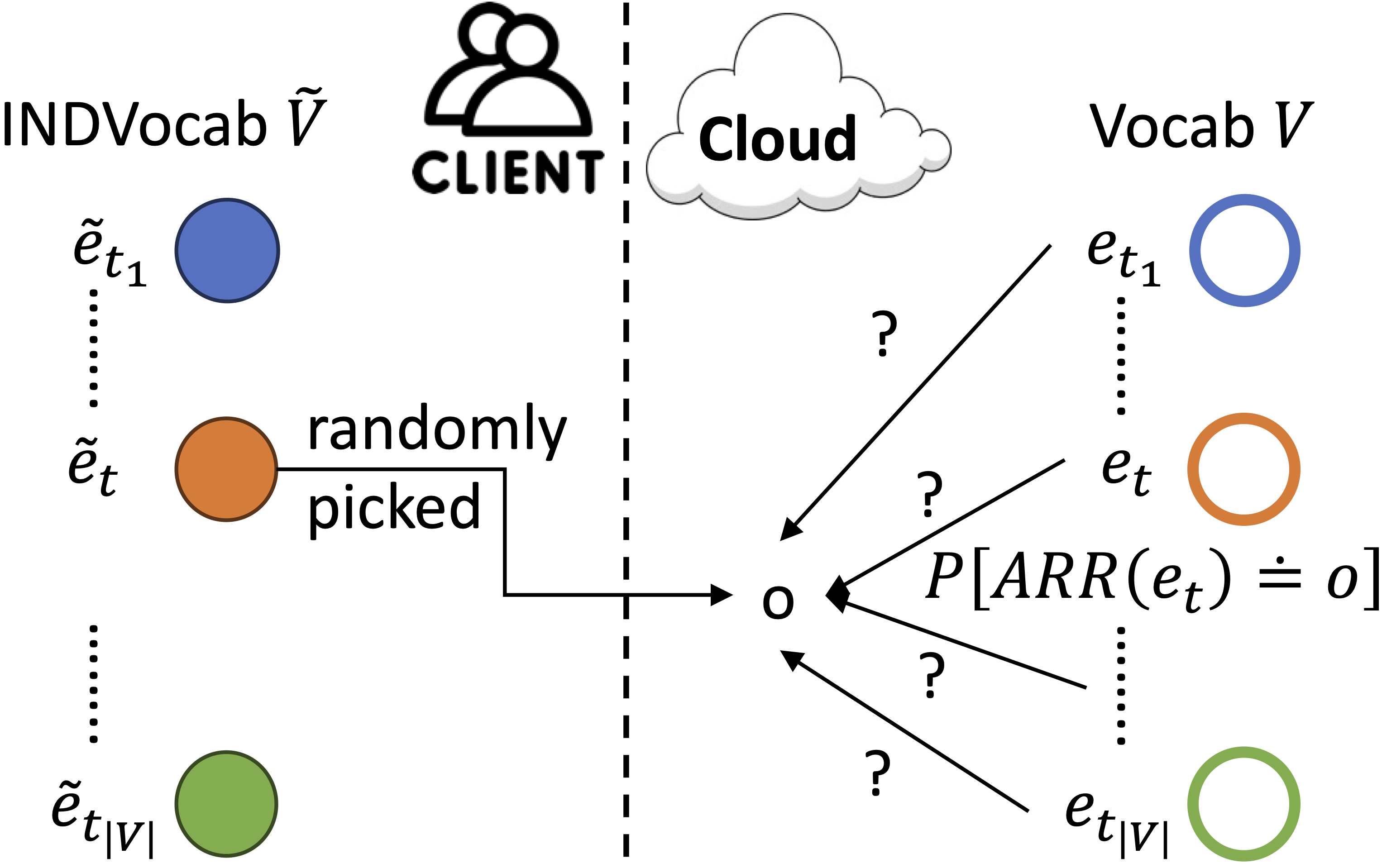} \vspace{-20pt}
  \end{center}
  \caption{Token Reconstruction Security Game.} \vspace{-10pt}
\label{ReconstructionGame}
\end{wrapfigure}

To achieve our goal, we introduce a security game (Figure \ref{ReconstructionGame}) where the client selects an $\epsilon$-IND-preserving token embedding $o \in \{\tilde{e}_t\}_{t \in V}$ and sends it to the cloud along with the original vocabulary $V$ associated with the ground-truth token embeddings $\{e_t\}_{t\in V}$. The cloud wins if it correctly identifies the ground-truth token $t$ corresponding to $o$. The cloud receives no feedback on its success, and the client does not provide additional information derived from the original embeddings, adhering to the CPC constraints. The frequency of game play or the selection of $\epsilon$-IND-preserving embeddings does not compromise the IND protection due to its post-processing property \citep{dwork2014}. Each game remains independent, as the cloud gains no insight into the \textsc{INDVocab} $\tilde{V}$ from participating in multiple games. Theorem \ref{upper-bound token} quantifies the cloud's probability of correctly inferring $t$ from $o$, $P[ARR(e_t) \doteq o]$. Note that $\forall t \neq t': \tilde{e}_t \neq \tilde{e}_{t'}$.

\begin{theorem} Given an $\epsilon$-\textsc{INDVocab} $\tilde{V}$, the probability that an arbitrary $t \in V$ is the ground-truth token of an observed $\epsilon$-IND token embedding $o \in \{\tilde{e}_t\}_{t \in V}$, denoted as $P[ARR(e_t) \doteq o]$, is bounded as follows:
\begin{footnotesize}
\begin{equation}
\forall t \in V: \frac{1}{1 + (|V|-1)e^\epsilon} \leq P[ARR(e_t) \doteq o] \leq \frac{e^\epsilon}{e^\epsilon + |V| - 1}.
\label{bounded token}
\end{equation}
\end{footnotesize}
\label{upper-bound token}
\end{theorem}

\begin{proof}
Given an arbitrary $\epsilon$-IND-preserving token embedding $o \in \{\tilde{e}_t\}_{t \in V}$ and the vocabulary $V$, the cloud will make an inference whether $o$ is a result of randomizing one of the $|V|$ token embeddings in the vocabulary. Therefore, given $|V|$ possible outcomes from the cloud perspective, we have
\begin{small}
\begin{align}
&\sum_{t \in V} P[ARR(e_t) \doteq o] = 1 \label{eq: sum 1} \\
\Leftrightarrow & \forall t \in V: P[ARR(e_t) \doteq o] = 1 - \sum_{t' \in V \setminus t} P[ARR(e_{t'}) \doteq o]. \nonumber
\end{align}
\end{small}

From Definition \ref{LDPVocab} and Theorems \ref{theorem-beta-bound}-\ref{Theorem-composition}, we also have that: $\forall t' \in V \setminus t, \forall \mathcal{O} = \{z^i \in Range(ARR)\}_{\forall i^{\text{th}}\text{-feature in } e_t}: P[ARR(e_{t'}) = \mathcal{O}] \geq \frac{P[ARR(e_t) = \mathcal{O}]}{e^\epsilon} \Leftrightarrow -P[ARR(e_{t'}) = \mathcal{O}] \leq - \frac{P[ARR(e_t) = \mathcal{O}]}{e^\epsilon}$. Given the observed outcome $o \in \mathcal{O}$, we have the that: $-P[ARR(e_{t'}) \doteq o] \leq - \frac{P[ARR(e_t) \doteq o]}{e^\epsilon}$. By applying this result to all $t' \in V \setminus t$, we have that
\begin{small}
\begin{equation}
-\sum_{t' \in V \setminus t} P[ARR(e_{t'}) \doteq o] \leq -(|V| - 1)\frac{P[ARR(e_t) \doteq o]}{e^\epsilon}.
\label{eq: sum 2}
\end{equation}
\end{small}

From Eqs. \ref{eq: sum 1} and \ref{eq: sum 2}: we have that $\forall t \in V: P[ARR(e_t) \doteq o] \leq 1 -(|V| - 1)\frac{P[ARR(e_t) \doteq o]}{e^\epsilon}$
\begin{small}
\begin{align} 
\Leftrightarrow P[ARR(e_t) \doteq o] \leq \frac{1}{1 + \frac{|V| - 1}{e^\epsilon}} = \frac{e^\epsilon}{e^\epsilon + |V| - 1}.
\label{upperbound token}
\end{align}
\end{small}

In addition, we also have that $\forall t' \in V \setminus t, \forall \mathcal{O} = \{z^i \in Range(ARR)\}_{\forall i^{\text{th}}\text{-feature in } e_t}: P[ARR(e_{t'}) = \mathcal{O}] \leq e^\epsilon P[ARR(e_t) = \mathcal{O}]$. Given the observed outcome $o \in \mathcal{O}$, we have the that: $P[ARR(e_{t'}) \doteq o] \leq e^\epsilon P[ARR(e_t) \doteq o]$. By applying this result to all $t' \in V \setminus t$, we have that 
\begin{small}
\begin{equation}
1 - \sum_{t' \in V \setminus t} P[ARR(e_{t'}) \doteq o] \geq 1 - (|V| - 1) e^\epsilon P[ARR(e_t) \doteq o].
\label{lowerbound 1}
\end{equation}
\end{small}

From Eqs. \ref{eq: sum 1}, \ref{lowerbound 1}, we have $P[ARR(e_t) \doteq o] \geq 1 - (|V| - 1) e^\epsilon P[ARR(e_t) \doteq o]$
\begin{small}
\begin{align}
\Leftrightarrow P[ARR(e_t) \doteq o] \geq \frac{1}{1 + (|V|-1)e^\epsilon}.
\label{lowerbound 2}
\end{align}
\end{small}

\noindent Consequently, from Eqs. \ref{upperbound token} and \ref{lowerbound 2}, Theorem \ref{upper-bound token} holds.
\end{proof}


From Theorem \ref{upper-bound token}, it is evident that the \textit{``larger the number of tokens''} in the vocabulary $V$ and the \textit{``smaller the privacy budget $\epsilon$''} is, the \textit{``lower the probability''} for the cloud to infer the ground-truth token, offering rigorous privacy protection. 


We can extend the security game by allowing the client to pick an arbitrary prompt $x = \{t_j\}_{j = 1}^{|x|}$ represented as a sequence of $\epsilon$-IND preserving tokens $\{o_j\}_{j = 1}^{|x|}$ s.t. $\forall j \in [1, |x|]: o_j = \tilde{e}_{t_j}$ and the client sends $\{o_j\}_{j = 1}^{|x|}$ to the cloud. The cloud wins the game if its reconstructed prompt $\hat{x}$ provides a clear gist of $x$ by identifying ground-truth tokens $\{t_j\}$ in $x$. For instance, $Bleu(\hat{x}, x) \geq \rho$ $(= 20)$ for a clear gist of $x$ \citep{BleuScoreRange}, formulated as
$P[Bleu(\hat{x}, x) \geq \rho; \{o_j\}_{j = 1}^{|x|}]$. Like the previous security game, the client does not send the game outcome to the cloud. Therefore, \textbf{the cloud observes no extra information} derived from the ground-truth token embeddings \textbf{in playing the games}. As a result, the number of games, the number of times a token appears in one or more prompts, and the number of times a prompt is selected in these games do not affect the $\epsilon$-IND protection of the token embeddings (the post-processing property of DP \citep{dwork2014}).

In this security game, the numbers of tokens in the reconstructed and ground-truth prompts, $\hat{x}$ and $x$, are the same: $|\hat{x}| = |x|$. Hence, Rouge precision is equal to Rouge recall: $Rouge\text{-}precision(\hat{x}, x), Rouge\text{-}recall(\hat{x}, x) = C / |x|$, where $C$ is the number of correct reconstructed tokens in $\hat{x}$. As a result, $Rouge\text{-}F1(\hat{x}, x)$ \citep{lin-2004-rouge} is equal to $C / |x|$, which is equivalent to $Bleu(\hat{x}, x) = C/|x|$. Therefore, we focus on analyzing the correlation between the IND budget $\epsilon$, the vocabulary size $|V|$, the size of the prompt $|x|$, and the threshold $\rho$ in reconstructing a clear gist of $x$ using the $Bleu$ score below.
The following proposition limits the cloud's upper-bounded probability of winning the security game.




\begin{prop} The cloud's probability to identify ground-truth tokens in $x$ with a gist level higher than or equal $\rho$ is upper-bounded as
$P\big[Bleu(\hat{x}, x) \geq \rho; \{o_j\}_{j = 1}^{|x|}\big] \leq \big(\frac{\psi e^\epsilon + 1}{\psi e^\epsilon + \psi^2}\big)^{\rho |x|} \times \big(\frac{\psi e^\epsilon}{\psi e^\epsilon + 1}\big)^{(1-\rho)|x|}, $ where $\psi = |V| - 1$.
\label{Bleu Bound}
\end{prop}

\begin{proof}
Let us recall the Bleu score definition as follows: $Bleu(\hat{x}, x) = BP(\hat{x}, x) \times C / |x|,$
where $BP(\hat{x}, x)$ stands for the brevity penalty given for the mismatched length between the reconstructed prompt $\hat{x}$ and ground-truth prompt $x$, and $C$ is the number of correct reconstructed tokens in $\hat{x}$.
In our security game, $BP(\hat{x}, x) = 1$ since $\hat{x}$ and $x$ have the same length. 

The probability of having $C$ correct reconstructed tokens, assuming $\{\hat{t}_j = t_j\}_{j = 1}^C$ and $\{\hat{t}_j \neq t_j\}_{j = C+1}^{|x|}$ without loss of generality, is:
\begin{small}
\begin{multline}
P\big(\{\hat{t}_j = t_j\}_{j = 1}^C, \{\hat{t}_j \neq t_j\}_{j = C+1}^{|x|}; \{o_j\}_{j = 1}^{|x|}\big) \\ = \prod_{j = 1}^C P[ARR(e_{t_j}) \doteq o_j] \times \prod_{j = C+1}^{|x|} P[ARR(e_{t_j}) \not\doteq o_j],
\label{eq: C probability}
\end{multline}
\end{small}

\noindent where $P[ARR(e_{t_j}) \not\doteq o_j]$ indicates the cloud's probability to infer than $t_j$ is not the ground-truth token of the observed $\epsilon$-IND preserving token embedding $o_j$.

From Theorem \ref{upper-bound token}, we have that $\prod_{j = 1}^C P[ARR(e_{t_j}) \doteq o_j] \leq (\frac{e^\epsilon}{e^\epsilon + |V| - 1})^C$. In addition, we have that
\begin{small}
\begin{align}
& P[ARR(e_{t_j}) \not\doteq o_j] = 1 - P[ARR(e_{t_j}) \doteq o_j] \\
& \leq 1 - \frac{1}{1 + (|V|-1)e^\epsilon} = \frac{(|V| - 1)e^\epsilon}{1 + (|V| - 1)e^\epsilon}.
\label{incorrect upperbound}
\end{align}
\end{small}

\noindent Let's denote $\psi = |V|-1$, from Eqs. \ref{bounded token}, \ref{eq: C probability}, and \ref{incorrect upperbound}, we have
\begin{footnotesize}
\begin{align} 
& P(\{\hat{t}_j = t_j\}_{j = 1}^C, \{\hat{t}_j \neq t_j\}_{j = C+1}^{|x|}; \{o_j\}_{j = 1}^{|x|}) \leq (\frac{e^\epsilon}{e^\epsilon + |V| - 1})^C \times \\
& \big(\frac{(|V|-1)e^\epsilon}{1 + (|V|-1)e^\epsilon}\big)^{|x| - C} = \big(\frac{\psi e^\epsilon + 1}{\psi e^\epsilon + \psi^2}\big)^C \times \big(\frac{\psi e^\epsilon}{\psi e^\epsilon + 1}\big)^{|x|-C},
\label{eq: C monotonic}
\end{align}
\end{footnotesize}

From Eq. \ref{eq: C monotonic}, the probability to reconstruct a clearer gist of a prompt $x$, i.e., $C$ increases, is reduced monotonically. It is because $\psi \gg 1$ and $C \in [0, |x|]$. Therefore, we have that: $P\big[Bleu(\hat{x}, x) \geq \rho; \{o_j\}_{j = 1}^{|x|}\big]  = P\big[C \geq \rho |x|; \{o_j\}_{j = 1}^{|x|}\big] \leq P\big[C = \rho |x|; \{o_j\}_{j = 1}^{|x|}\big]
= P(\{\hat{t}_j = t_j\}_{j = 1}^{\rho |x|}, \{\hat{t}_j \neq t_j\}_{j = \rho |x|+1}^{|x|}; \{o_j\}_{j = 1}^{|x|}) \leq \big(\frac{\psi e^\epsilon + 1}{\psi e^\epsilon + \psi^2}\big)^{\rho |x|} \times \big(\frac{\psi e^\epsilon}{\psi e^\epsilon + 1}\big)^{(1-\rho)|x|}.$

Consequently, Proposition \ref{Bleu Bound} holds.
\end{proof}

\begin{theorem} The cloud's previously reconstructed token sequences $\hat{t}_{<j}$ can enhance its probability of correctly reconstructing the next token $t_j$: $P(\hat{t}_j = t_j | \hat{t}_{<j})$. This advantage is bounded by a constant $\gamma$ in practice, as follows: 
\begin{small}
\begin{equation}
\forall t_j \in x: 0 \leq P(\hat{t}_j = t_j | \hat{t}_{<j}) - P(\hat{t}_j = t_j) \le \gamma, \text{\ where } \gamma \in [0, 1].
\label{advantages token sequences}
\end{equation}
\end{small}
Given a bounded constant $\gamma$, we extend Proposition 1, bounding the probability to reconstruct a gist level higher than $\rho$ of the prompt $x$ exploiting $\hat{t}_{<j}$, as follows:
\begin{footnotesize}
\begin{equation}
P\big[Bleu(\hat{x}, x) \geq \rho; \{o_j\}_{j = 1}^{|x|}\big] \\ \leq \big(\frac{\psi e^\epsilon + 1}{\psi e^\epsilon + \psi^2} + \gamma \big)^{\rho |x|} \times \big(\frac{\psi e^\epsilon}{\psi e^\epsilon + 1} - \gamma\big)^{(1-\rho)|x|}. \nonumber
\end{equation}
\end{footnotesize}
\label{Token sequences bound}
\end{theorem}
\begin{proof} The probability of having $C$ correct reconstructed tokens, assuming $\{\hat{t}_j = t_j|\hat{t}_{<j}\}_{j = 1}^C$ and $\{\hat{t}_j \neq t_j|\hat{t}_{<j}\}_{j = C+1}^{|x|}$ without loss of generality, is: $P\big[Bleu(\hat{x}, x) \geq \rho; \{o_j\}_{j = 1}^{|x|}\big]$
\begin{small}
\begin{align}
& = P\big[C \geq \rho |x|; \{o_j\}_{j = 1}^{|x|}\big] \leq P\big[C = \rho |x|; \{o_j\}_{j = 1}^{|x|}\big] \\
& = P(\{\hat{t}_j = t_j|\hat{t}_{<j}\}_{j = 1}^{\rho |x|}, \{\hat{t}_j \neq t_j |\hat{t}_{<j}\}_{j = \rho |x|+1}^{|x|}; \{o_j\}_{j = 1}^{|x|}). \label{mid-probability}
\end{align}
\end{small}
From Eq. \ref{advantages token sequences}, we have that $\forall j: P(\hat{t}_j = t_j | \hat{t}_{<j}) \le P(\hat{t}_j = t_j) + \gamma$. Let us consider the worst-case for the client, in which the cloud has the maximal probability of correctly reconstructing every token $t_j$, denoted as $P^*(\hat{t}_j = t_j | \hat{t}_{<j}) = P(\hat{t}_j = t_j) + \gamma$. The corresponding probability of incorrectly reconstructing a token $t_j$ is: $P(\hat{t}_j \neq t_j | \hat{t}_{<j}) = 1 - P^*(\hat{t}_j = t_j | \hat{t}_{<j})$. From Eq. \ref{mid-probability},
\begin{footnotesize}
\begin{align}
& P\big[Bleu(\hat{x}, x) \geq \rho; \{o_j\}_{j = 1}^{|x|}\big] \leq \prod_{j =1}^{\rho|x|} P^*(\hat{t}_j = t_j | \hat{t}_{<j}) \prod^{|x|}_{j = \rho|x|+1} P(\hat{t}_j \neq t_j | \hat{t}_{<j}) \nonumber \\
& = \big(\frac{\psi e^\epsilon + 1}{\psi e^\epsilon + \psi^2}+\gamma\big)^{\rho |x|} \times \big(\frac{\psi e^\epsilon}{\psi e^\epsilon + 1}-\gamma\big)^{(1-\rho)|x|}
\end{align}
\end{footnotesize}
Consequently, Theorem \ref{Token sequences bound} holds.
\end{proof}

\section{Complexity and Cost Analysis}
\label{Complexity and Cost Analysis}

We evaluate the complexity and cost of operating \textsc{NOIR}. Let $n$ and $d$ be the numbers of input tokens and the hidden size of attention layers. The network communication between the client and cloud involves two phases: sending the encoder's output embedding $\mathcal{E}$ (complexity $O(nd)$) and receiving the enriched embedding $\ddot{\mathcal{E}}$ (cost $O(nd)$), resulting in a total transmission cost of $O(nd)$. On the client side, the computation consists of three stages: generating the privatized token embedding ($O(n)$), and feeding embeddings through the encoder and decoder, each with complexity $O(n^2d)$, leading to a total cost of $O(n^2d)$ (comparable with \citep{NIPS2017_3f5ee243}). Fine-tuning costs are notably reduced, as the client computes the decoder's gradients twice and the encoder's gradients once per training iteration, compared to computing gradients for the entire LLM. In fact, using 1 and 4 attention blocks for the encoder and decoder, respectively, in CodeLlama-7B results in 71.9\% lower tuning costs. On the cloud side, prompting and hosting costs are reduced because the prompt embedding passes through fewer attention blocks (27 vs. 32), lowering the cloud prompting and hosting cost by 15.6\%. \textsc{NOIR} is scalable, requiring minimal resources for fine-tuning/operation.

\section{Pilot General Text Generation Pipeline}
\label{General Text}

To adapt \textsc{NOIR} for general text generation, we propose a six-phase fine-tuning strategy mirroring pre-training while escalating reasoning complexity. \textbf{Phase 1} trains the encoder, decoder, and middle block's LoRA on simple math/logic datasets to generate explicit reasoning steps. \textbf{Phase 2} integrates multi-hop and scientific reasoning datasets to synthesize coherent reasoning from multiple facts. \textbf{Phase 3} strengthens problem-solving across domains (math, code, writing) using high-quality datasets for deeper abstract reasoning. \textbf{Phase 4} extends handling of long contexts (up to 32,000 tokens) and multi-turn conversations with datasets like patents and dialogues, enabling interactive reasoning and plan revision. \textbf{Phase 5} employs Group Sequence Policy Optimization (GSPO) \cite{zheng2025group} to align reasoning sequences, ensuring structured, repeat-free outputs with correct final answers. \textbf{Phase 6} tailors the model to enterprise use cases via domain-specific corpora while retaining cross-domain capabilities. A small subset of prior data is retained in all phases to prevent forgetting. We deployed a Qwen3-32B with $\epsilon=27$ (\url{https://noir.oppyai.com}) and on-going trials with Qwen3-235B-A22B-Instruct using this training pipeline. The curated dataset has total of $\sim$2.7m data points across six phases. We achieve highly competitive performance across reasoning tasks on the LiveBench \cite{white2024livebench} while maintaining $\epsilon$-IND (Table \ref{LiveBenchPerformance}). Note that, poor data quality and fine-tuning pipeline can significantly lower \textsc{NOIR}'s performance rendering it unpractical for enterprises.

\begin{table}[h]
\footnotesize
\centering
\begin{tabular}{lcc}
\hline
\textbf{} & \textbf{NOIR} & \textbf{Qwen3-32B} \\
\hline
Reasoning Average & 66.5 & 48.2 \\
Mathematics Average & 71.8 & 67.4 \\
Instruction Following Average & 70.2 & 17.8 \\
Language Average & 41.3 & 55.0 \\
\hline
\end{tabular}
\caption{\textsc{NOIR}'s Performance on LiveBench-2026-01-08. Qwen3-32B's Performance is Recorded on the Leader Board.}
\label{LiveBenchPerformance}
\end{table}

\section{Preliminary Results on Proprietary LLMs}
\label{NVIBPrelim}

To adapt NOIR for black-box middle blocks from proprietary LLMs, we introduce a representation-learning module that maps both sides into a shared, pre-agreed embedding space, enabling seamless integration. This approach extends the Nonparametric Variational Information Bottleneck \cite{henderson2023a} (NVIB) autoencoder, which offers attention-like expressivity, model-agnostic flexibility, and theoretical rigor, making it ideal for bridging heterogeneous LLM components.


\begin{wrapfigure}{r}{0.65\columnwidth}
  \begin{center} \vspace{-17.5pt}
\includegraphics[width=0.65\columnwidth]{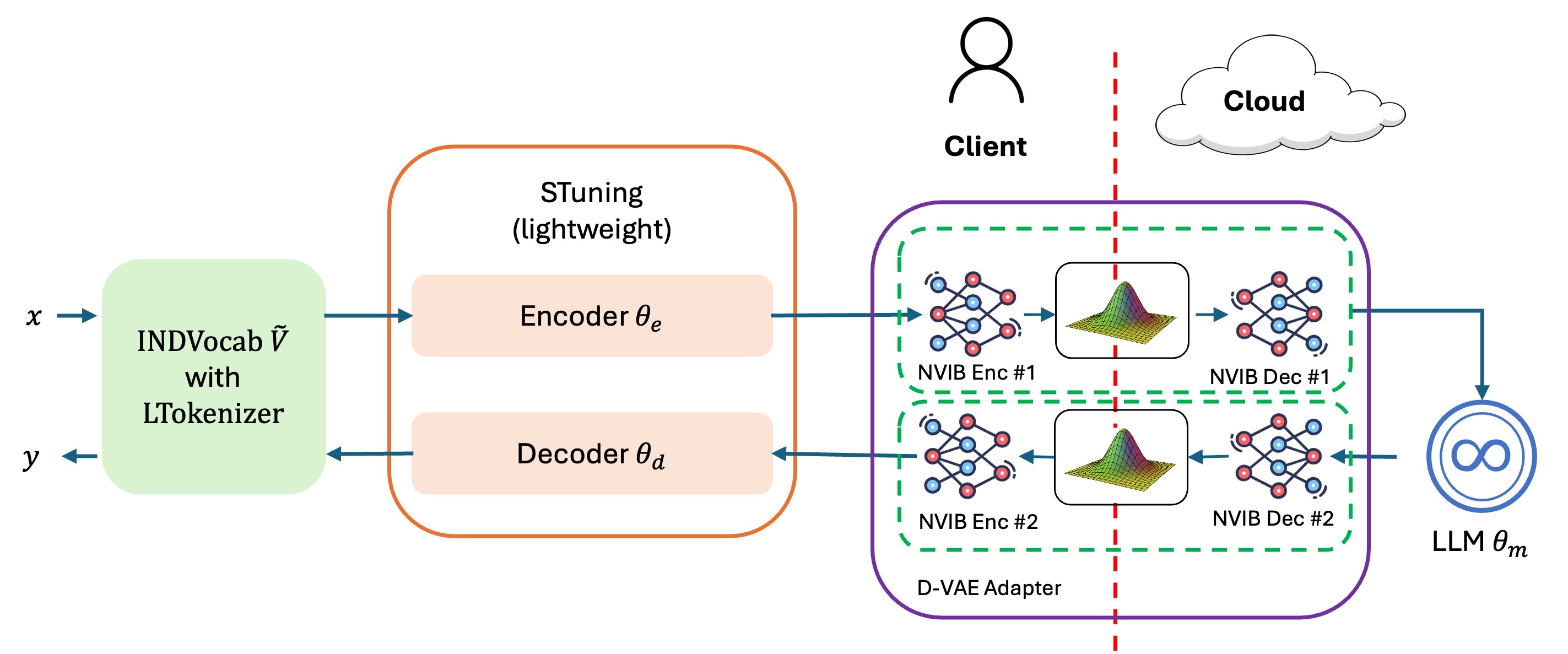} \vspace{-25pt}
  \end{center}
  \caption{D-VAE structure for encoder/decoder alignment.} \vspace{-7.5pt}
\label{img:nvib}
\end{wrapfigure}

Figure \ref{img:nvib} shows the D-VAE adapter block, which consists of two paired NVIB autoencoders. One autoencoder processes client-to-cloud representations, while the other handles cloud-to-client responses. Each NVIB encoder compresses input hidden states into a shared latent space with a pre-defined dimensionality, ensuring dimensional alignment and distribution compatibility across heterogeneous models. The corresponding decoders reconstruct the latent variables into embeddings that match the expected interface of the next component.

\begin{wrapfigure}{r}{0.50\columnwidth}
  \begin{center} \vspace{-15pt}
\includegraphics[width=0.50\columnwidth]{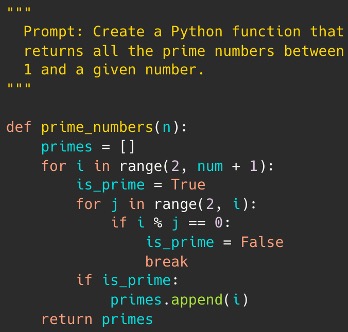} \vspace{-20pt}
  \end{center}
  \caption{An example of generated code and prompt.} \vspace{-5pt}
\label{img:nvib_examples}
\end{wrapfigure}

To evaluate our approach, we instantiate the client-side components using Qwen2.5-Coder-7B-Instruct, selecting the first attention block as the encoder and the last as the decoder for \textsc{NOIR}. The cloud-side middle block is set as the entire attention stack from LLaMA-3.1-8B, treated as a fixed black-box module. We train the end-to-end system for three epochs under two data regimes: a small-scale setting with 18k samples from the CodeAlpaca dataset and a large-scale setting with 100k samples from our curated dataset.


Our preliminary results show that the D-VAE aligns the two model structures. In the 18k CodeAlpaca setting, the system achieves a Pass@1 of 2.8\% on MBPP, indicating that limited data suffices for learning a functional cross-model interface. Scaling to 100k curated examples significantly boosts performance to 12.4\% Pass@1, underscoring the D-VAE adapter's ability to leverage larger datasets and enhance downstream code generation despite the heterogeneity and black-box nature of the intermediate LLM. Although the preliminary performance is far from practically usable, Figure \ref{img:nvib_examples} shows that the architecture generates coherent, syntactically valid, and task-appropriate code solutions, highlighting its potential and suggesting further improvements with scaling and refinement.

\end{document}